\documentclass[runningheads]{llncs}
\usepackage{amsmath}%
\usepackage{amssymb}%
\usepackage{stmaryrd}%
\usepackage{array}%
\usepackage{epic}%
\usepackage[fixamsmath]{mathtools}
\usepackage{cancel}
\usepackage{url}
\usepackage{bigstrut}
\usepackage{pifont}%
\usepackage{wrapfig}

\DeclareMathAlphabet\mathbb{U}{msb}{m}{n}
\DeclareMathAlphabet{\mathrm}    {OT1}{cmr}{m}{n}
\DeclareMathAlphabet{\mathrmbf}  {OT1}{cmr}{bx}{n}
\DeclareMathAlphabet{\mathrmit}  {OT1}{cmr}{m}{it}
\DeclareMathAlphabet{\mathrmbfit}{OT1}{cmr}{bx}{it}
\DeclareMathAlphabet{\mathsf}    {OT1}{cmss}{m}{n}
\DeclareMathAlphabet{\mathsfbf}  {OT1}{cmss}{bx}{n}
\DeclareMathAlphabet{\mathsfit}  {OT1}{cmss}{m}{sl}
\DeclareMathAlphabet{\mathtt}    {OT1}{cmtt}{m}{n}
\DeclareMathAlphabet{\mathttbf}  {OT1}{cmtt}{bx}{n}
\DeclareMathAlphabet{\mathttit}  {OT1}{cmtt}{m}{it}
\DeclareMathAlphabet{\mathpzc}   {OT1}{pzc}{m}{it}
\newcommand{\keywords}[1]{\par\addvspace\baselineskip\noindent\enspace\ignorespaces{\bfseries Keywords:\,}#1}
\newcommand{\comment}[1]{}
\begin{document}

\pagestyle{headings}
\title{The {\ttfamily ERA} of {\ttfamily FOLE}: Foundation} 
\titlerunning{{\ttfamily FOLE-ERA} Foundation}  
\author{Robert E. Kent}
\institute{Ontologos}
\maketitle

\begin{abstract}
This paper discusses the representation of ontologies
in the first-order logical environment {\ttfamily FOLE}.
An ontology defines the primitives 
with which to model the knowledge resources for a community of discourse.
These primitives consist of classes, relationships and properties.
An ontology uses formal axioms to constrain the interpretation of these primitives. 
In short, an ontology specifies a logical theory.
%
%
This paper continues the discussion 
of the representation 
and interpretation
of ontologies
in the first-order logical environment {\ttfamily FOLE}.
%
%
%
The formalism and semantics of (many-sorted) first-order logic 
can be developed in both
a \emph{classification form}
and
an 
\emph{interpretation form}.
%
%
Two papers,
the current paper,
defining the concept of a structure,
and 
``The {\ttfamily ERA} of {\ttfamily FOLE}: Superstructure'',
defining the concept of a sound logic,
represent 
the 
\emph{classification form},
corresponding to ideas discussed in the ``Information Flow Framework''.
Two papers, 
``The {\ttfamily FOLE} Table'',
defining the concept of a relational table,
and 
``The {\ttfamily FOLE} Database'',
defining the concept of a relational database,
represent 
the \emph{interpretation form},
expanding on material found in the paper 
``Database Semantics''.
%
%
%
%
Although 
the classification form 
follows the entity-relationship-attribute data model of Chen,
the interpretation form incorporates the relational data model of Codd.
A fifth paper
``{\ttfamily FOLE} Equivalence''
proves that
the classification form 
is equivalent to
the interpretation form.
%
In general,
the {\ttfamily FOLE} representation uses a conceptual structures approach,
that is completely compatible with 
the theory of institutions, 
formal concept analysis 
and information flow.
%
%
\end{abstract}

\tableofcontents


\newpage
\section{Introduction}\label{intro}

\subsection{Philosophy}\label{subsec:philosophy}



A conceptual model for a community 
represents the information needed by the community
---
the content, relationships and constraints 
necessary to describe the community.
The content consists of 
things of significance to the community (entities), and
characteristics of those things (attributes).
The relationships are associations between those things.
The entities are the core concepts 
that are used for representing the semantics of the community.
Entities are described by attributes, 
which are the various 
properties, characteristics, modifiers, aspects or features of entities. 
%
Hence,
the entity-relationship-attribute ({\ttfamily ERA}) formalism
is a ternary representation for knowledge,
since it uses three kinds for representation: entities, attributes and relations.
In contrast,
the first-order logical environment {\ttfamily FOLE} (Kent~\cite{kent:iccs2013})
followed the knowledge representation approach of
traditional many-sorted first-order logic ({\ttfamily MSFOL}).
Both the original 
{\ttfamily FOLE} formalism and the {\ttfamily MSFOL} formalism
are binary representations for knowledge,
since they use two kinds for representation: entities and relations.

%

However,
the first-order logical environment {\ttfamily FOLE}
can very naturally represent the {\ttfamily ERA} data model.
The idea is that the original {\ttfamily FOLE} relation represented a nexus of roles,
where the roles were played by the original {\ttfamily FOLE} entities.
In order to represent the {\ttfamily ERA} data model,
we think of the original {\ttfamily FOLE} relations as the new {\ttfamily FOLE} entities
described by a nexus of features or aspects,
where the aspects are represented by the new {\ttfamily FOLE} attributes,
which replace the original {\ttfamily FOLE} entities.
In the {\ttfamily FOLE} representation of the {\ttfamily ERA} data model,
entities and their attributes are primary notions,
whereas relationships are secondary notions that are subsumed by other constructs.
%
Some relations
(foreign keys, subtypes, sums, ...) 
have a special representation in {\ttfamily FOLE};
whereas, other relations can be resolved into concepts (entities) with a nexus of roles.

Tbl.~\ref{tbl:fole:era} shows the terminological correspondence between 
the basic components of (old/new) {\ttfamily FOLE} and {\ttfamily ERA}.
For example,
the original {\ttfamily FOLE} entity type
is renamed the new {\ttfamily FOLE} attribute type (sort),
and this corresponds to the {\ttfamily ERA} attribute type (data type);
and
the original {\ttfamily FOLE} relation instance
is renamed the new {\ttfamily FOLE} entity instance (key),
and this corresponds to the {\ttfamily ERA} entity.
\begin{table}
\begin{center}
\begin{tabular}{|r@{\hspace{5pt}$\simeq$\hspace{5pt}}l@{\hspace{5pt}$\sim$\hspace{5pt}}l|}
\hline
\multicolumn{1}{|l}{{\ttfamily FOLE} (old)}
&
\multicolumn{1}{l}{{\ttfamily FOLE} (new)}
&
\multicolumn{1}{l|}{{\ttfamily ERA}}
\\\hline\hline
relation
&
entity
&
entity
\\
entity
&
attribute
&
attribute
\\\hline
\end{tabular}
\end{center}
\caption{{\ttfamily FOLE}-{\ttfamily ERA} Correspondence}
\label{tbl:fole:era}
\end{table}
%

%
%

As commonly observed,
an entity is a thing capable of an independent existence that can be uniquely identified.
In natural language, an entity corresponds to a noun.
A relationship links entities,
and corresponds to a verb in natural language. 
Entities and relationships can both have attributes.  
In natural language, a relational attribute corresponds to a role or case.
Inclusion and subtype relationships are special kinds of relationships.
A data model can be visualized in terms of entities, relationships and attributes.
But in general,
relationships can be conceptualized by being converted to entities.

Hence,
a data model is more 
simply conceptualized in terms of entities and attributes.
When doing so,
there is an implied boundary around the conceptualization,
which converts an entity's collection of attributes
into a list
(possibly infinite in size or arity); 
a signature is the list of attribute types (sorts) associated with an entity type,
whereas
a tuple is the list of attribute instances (values) associated with an entity instance (key).
Entity types can be mapped to the associated signature, and
entity instances (primary keys) identify and can be mapped to the associated tuple
(horizontal dimension of Fig.~\ref{fig:fole:data:model}).
In general,
types classify instances.
Hence,
entity types classify keys 
and
sorts 
classify values 
(vertical dimension of Fig.~\ref{fig:fole:data:model}).
Implicit from the 
{\ttfamily ERA} data model
is an entity type system and multi-sorted logic,
which uses boolean operators and quantification,
and is defined in terms of signature-based fibers of formulas (queries) in Kent~\cite{kent:fole:era:supstruc}.

\begin{figure}
\begin{center}
\begin{tabular}{c}
\setlength{\unitlength}{0.55pt}\begin{picture}(120,120)(0,-20)
\put(0,98){\makebox(0,0){\footnotesize{$\mathrmbf{ent}$}}}
\put(120,98){\makebox(0,0){\footnotesize{$\mathrmbf{attr}$}}}
\put(-10,80){\makebox(0,0)[r]{\footnotesize{$\mathrmbf{typ}$}}}
\put(-10,0){\makebox(0,0)[r]{\footnotesize{$\mathrmbf{inst}$}}}
\put(0,0){\circle*{4.5}}
\put(120,0){\circle*{4.5}}
\put(0,80){\circle*{4.5}}
\put(120,80){\circle*{4.5}}
\put(5,10){\makebox(0,0)[l]{\scriptsize{$\mathrmit{key}$}}}
\put(115,70){\makebox(0,0)[r]{\scriptsize{$\mathrmit{sort}$}}}
\put(115,10){\makebox(0,0)[r]{\scriptsize{$\mathrmit{value}$}}}
\put(195,40){\makebox(0,0){\footnotesize{$\left.\rule{0pt}{27pt}\right\}\mathrmbfit{classification}$}}}
\put(60,-25){\makebox(0,0){\footnotesize{$\overset{\underbrace{\rule{70pt}{0pt}}}{\mathrmbfit{hypergraph}}$}}}
\thicklines
\put(5,80){\line(1,0){110}}
\put(5,0){\line(1,0){110}}
\put(0,5){\line(0,1){70}}
\put(120,5){\line(0,1){70}}
\put(69,80){\vector(1,0){0}}
\put(69,0){\vector(1,0){0}}
\end{picture}
\end{tabular}
\end{center}
\caption{{\ttfamily ERA} Data Model in {\ttfamily FOLE}}
\label{fig:fole:data:model}
\end{figure}
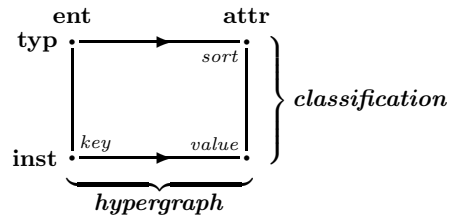
%

%
%

In review,
the simplest way to handle {\emph{things}} is
{\bfseries{first}} to distinguish {\emph{types}} from {\emph{instances}}
along the {\ttfamily FOLE} classification dimension in Fig.~\ref{fig:fole:data:model},
and
{\bfseries{second}} to view things (either types or instances) as participating 
in Whitehead's fundamental {\underline{prehension}} relationship
(Sowa~\cite{sowa:kr})
along the {\ttfamily FOLE} hypergraph dimension in Fig.~\ref{fig:fole:data:model},
which links 
a prehending thing called an {\emph{entity}} 
to
a prehended thing called an {\emph{attribute}}:
``an entity has an attribute''.
%
The {\ttfamily ERA} data model of {\ttfamily FOLE} uses an inclusive prehension for things;
hence, 
it is a mixed data model;
some entities are not attributes,
some attributes are not entities,
and some are both.
For any type in the overlap $\mathcal{E}\cap\mathcal{A}$,
any instance of that type is also in the overlap.
Foreign keys
are examples of
things that are both entities and attributes,
things in the overlap.


%
\begin{center}
{{\begin{tabular}{@{\hspace{-10pt}}c@{\hspace{15pt}}c@{\hspace{10pt}}c@{\hspace{10pt}}c}
\begin{tabular}{r@{\hspace{5pt}:\hspace{5pt}}l}
{\footnotesize{$\mathcal{E}$}} & entities
\\
{\footnotesize{$\mathcal{A}$}} & attributes
\end{tabular}
&
\begin{tabular}{|c|}\hline
\setlength{\unitlength}{0.6pt}
\begin{picture}(140,80)(0,0)
\put(30,40){\circle{60}}
\put(100,40){\circle{60}}
\put(30,0){\makebox(0,0){\footnotesize{$\mathcal{E}$}}}
\put(100,0){\makebox(0,0){\footnotesize{$\mathcal{A}$}}}
\put(65,-17){\makebox(0,0){\scriptsize{$\mathcal{E}{\,\cap\,}\mathcal{A}=\emptyset$}}}
\end{picture}
\\\\\hline
\end{tabular}
&
\begin{tabular}{|c|}\hline
\setlength{\unitlength}{0.6pt}
\begin{picture}(100,80)(0,0)
\put(30,40){\circle{60}}
\put(60,40){\circle{60}}
\put(30,0){\makebox(0,0){\footnotesize{$\mathcal{E}$}}}
\put(60,0){\makebox(0,0){\footnotesize{$\mathcal{A}$}}}
\put(45,-17){\makebox(0,0){\scriptsize{$\mathcal{E}{\,\cap\,}\mathcal{A}\not=\emptyset$}}}
\end{picture}
\\\\\hline
\end{tabular}
&
\begin{tabular}{|c|}\hline
\setlength{\unitlength}{0.6pt}
\begin{picture}(80,80)(0,0)
\put(40,40){\circle{60}}
\put(25,0){\makebox(0,0){\footnotesize{$\mathcal{E}$}}}
\put(40,0){\makebox(0,0){\footnotesize{\bfseries{=}}}}
\put(55,0){\makebox(0,0){\footnotesize{$\mathcal{A}$}}}
\end{picture}
\\\\\hline
\end{tabular}
\\ &&& \\
&{\itshape{Disjoint Model}}&{\itshape{Mixed Model}}&{\itshape{Unified Model}}
\end{tabular}}}
\end{center}
%


The author's ``Systems Consequence'' paper (Kent~\cite{kent:iccs2009}) 
is a very general theory and methodology for specification and inter-operation of systems of information resources. 
The generality comes from the fact that it is independent of the logical/semantic system (institution) being used. 
This is a wide-ranging theory, 
based upon ideas from 
information flow (Barwise and Seligman~\cite{barwise:seligman:97}), 
formal concept analysis (Ganter and Wille et al~\cite{ganter:wille:99}), 
the theory of institutions (Goguen et al~\cite{goguen:burstall:92}), and 
the lattice of theories notion (Sowa~\cite{sowa:kr}), 
for the integration of both formal and semantic systems independent of logical environment. 
In order to better understand the motivations of that paper and to be able more readily to apply its concepts, 
in the future it will be important to study system consequence in various particular logical/semantic systems. 
This paper aims to do just that for the logical/semantic system of relational databases. 
The paper, 
which was inspired by and which extends a recent set of papers on the theory of relational database systems 
(Spivak~\cite{spivak:sd},\cite{spivak:fdm}), 
is linked with work on the Information Flow Framework (IFF~\cite{iff}) connected with the ontology standards effort (SUO), 
since relational databases naturally embed into first order logic. 
We offer both an intuitive and a technical discussion. 
Corresponding to the notions of primary and foreign keys, 
relational database semantics takes two forms: 
a distinguished form where entities are distinguished from relations, and 
a unified form where relations and entities coincide. 
The distinguished form corresponds to the theory presented 
in the paper (Spivak~\cite{spivak:sd}). 
We extend Spivak's treatment of tables 
from the static case of a single entity classification (type specification) 
to the dynamic case of classifications varying along infomorphisms. 
Our treatment of relational databases as diagrams of tables differs from 
Spivak's sheaf theory of databases. 
The unified form, 
a special case of the distinguished form, 
corresponds to the theory presented in the paper (Spivak~\cite{spivak:fdm})). 
The unified form has a graphical presentation, 
which corresponds to the sketch theory of databases (Johnson and Rosebrugh~\cite{johnson:rosebrugh:07}) 
and the resource description framework (RDF). 
This paper, 
which is the first step to connect relational databases with system consequence, 
is concerned with the semantics of relational databases. 
Later papers will discuss various formalisms of relational databases, 
such as 
first order logic
(Kent~\cite{kent:fole:era:supstruc})
 and 
relational algebra
(Kent~\cite{kent:fole:rel:ops}
.

%
\newpage
\subsection{Knowledge Representation}

Many-sorted (multi-sorted) first-order predicate logic 
represents a community's ``universe of discourse'' as 
a heterogeneous collection of objects
by conceptually scaling
the universe according to types.
The \emph{relational model} (Codd~\cite{codd:90}) 
is an approach for the information management of a ``community of discourse''
\footnote{Examples include:
an academic discipline;
a commercial enterprise;
library science;
the legal profession;
etc.}
using the semantics and formalism of (many-sorted) first-order predicate logic. 
%
The relational model was initially discussed in two papers:
``A Relational Model of Data for Large Shared Data Banks''
by Codd \cite{codd:70} 
and
``The Entity-Relationship Model -- Toward a Unified View of Data'' 
by Chen \cite{chen:76}. 
The relational model follows many-sorted logic
by representing data in terms of many-sorted relations, 
subsets of the Cartesian product of multiple domains. 
All data is represented horizontally in terms of tuples, 
which are grouped vertically into relations. 
A database organized in terms of the relational model 
is a called relational database.
The relational model provides a method 
for modeling the data stored in a relational database 
and for defining queries upon it. 
%

\subsection{First Order Logical Environment}\label{sub:sec:fole}

\paragraph{Basics.}

The \emph{first-order logical environment} \texttt{FOLE}
is a category-theoretic representation for 
many-sorted (multi-sorted) first-order predicate logic. 
%
\footnote{Following the original discussion of {\ttfamily FOLE} (Kent~\cite{kent:iccs2013}), 
we use 
the term \emph{mathematical context} for the concept of a category,
the term \emph{passage} for the concept of a functor, and
the term \emph{bridge} for the concept of a natural transformation.
A context represents some ``species of mathematical structure''. 
A passage is a ``natural construction on structures of one species, 
yielding structures of another species'' 
(Goguen \cite{goguen:cm91}).}
%
The relational model can naturally be represented in \texttt{FOLE}.
The {\texttt{FOLE}} approach to logic, 
and hence to databases, 
relies upon two mathematical concepts:
(1) lists and (2) classifications.
Lists represent database signatures and tuples;
classifications represent data-types and logical predicates.
{\texttt{FOLE}} 
represents the header of a database table as a list of sorts, and
represents the body of a database table as a set of tuples 
classified by the header.
The notion of a list is common in category theory.
The notion of a classification is described in two books:
``Information Flow: The Logic of Distributed Systems''
by Barwise and Seligman \cite{barwise:seligman:97} and
''Formal Concept Analysis: Mathematical Foundations''
by Ganter and Wille \cite{ganter:wille:99}.

\paragraph{Architecture.}

%


A series of papers provides a rigorous mathematical basis for {\ttfamily FOLE} 
by defining 
an architectural semantics for the relational data model,
thus providing the foundation for
the formalism and semantics of first-order logical/relational database systems.
This architecture 
consists of two hierarchies of two nodes each:
the classification hierarchy
and
the interpretation hierarchy.

%

%

\begin{itemize}
%
\item
Two papers provide a precise mathematical basis for \texttt{FOLE} classification.
The current paper 
``The {\ttfamily ERA} of {\ttfamily FOLE}: Foundation'',
develops the notion of a \texttt{FOLE} \underline{\emph{structure}},
following the entity-relationship model of Chen~\cite{chen:76}.
This provides a basis for
the paper 
``The {\ttfamily ERA} of {\ttfamily FOLE}: Superstructure''
\cite{kent:fole:era:supstruc},
which develops the notion of a \texttt{FOLE} \underline{\emph{sound logic}}.
\newline
\item
Two papers provide a precise mathematical basis for \texttt{FOLE} interpretation.
Both of these papers expand on material found in the paper 
``Database Semantics''
\cite{kent:db:sem}.
The paper 
``The {\ttfamily FOLE} Table''
\cite{kent:fole:era:tbl},
develops the notion of a \texttt{FOLE} \underline{\emph{table}}
following the relational model of Codd~\cite{codd:90}.
This provided a basis for 
the paper 
``The {\ttfamily FOLE} Database''
\cite{kent:fole:era:tbl},
which develops the notion of a \texttt{FOLE} relational \underline{\emph{database}}.
\end{itemize}
%


%
\begin{flushleft}
{{\setlength{\extrarowheight}{1.6pt}
{{
{\begin{tabular}[t]{l@{\hspace{20pt}}l}
{{{\begin{minipage}{230pt}
The architecture of
{\ttfamily FOLE}
is pictured briefly
on the right
and more completely in
Fig.\,1
of the preface of
the paper
\cite{kent:fole:equiv}.
This consists of two hierarchies of two nodes each.
\comment{
\\
\vspace{-16pt}
\begin{description}
\item[{The classification hierarchy}]
on the left
defines
{\ttfamily FOLE} Structures
\cite{kent:fole:era:found}
at the bottom
and {\ttfamily FOLE} Sound Logics
\cite{kent:fole:era:supstruc}
at the top.
%
\item[\emph{The interpretation hierarchy}]
on the right
defines
{\ttfamily FOLE} Tables \cite{kent:fole:era:tbl}
at the bottom
and {\ttfamily FOLE} Databases 
(this paper)
at the top.
%
\end{description}
\vspace{-5pt}
}
The paper
``{\ttfamily FOLE} Equivalence''
\cite{kent:fole:equiv}
proves that
{\ttfamily FOLE} sound logics
are equivalent to
{\ttfamily FOLE} databases.
%
\end{minipage}}}}
&
{{\begin{tabular}{c@{\hspace{5pt}}}
\setlength{\unitlength}{0.36pt}
\begin{picture}(180,120)(-50,-20)
\put(60.5,80.5){\makebox(0,0){\tiny{$\equiv$}}}
%
\put(-60,93){\makebox(0,0){\tiny{\tt{Relational}}}}
\put(-60,73){\makebox(0,0){\tiny{\tt{Calculus}}}}
\put(180,10){\makebox(0,0){\tiny{\sf{Relational}}}}
\put(180,-10){\makebox(0,0){\tiny{\sf{Algebra}}}}
\qbezier(100,87)(60,97)(20,87)
\put(20,87){\vector(-4,-1){0}}
\qbezier(20,75)(60,65)(100,75)
\put(100,75){\vector(4,1){0}}
\put(0.3,80){\makebox(0,0){\huge{$\circ$}}}
\put(120.3,80){\makebox(0,0){\huge{$\circ$}}}
\put(1,10){\makebox(0,0){\huge{$\bullet$}}}
\put(122,10){\makebox(0,0){\huge{$\circ$}}}
%
\put(0,68){\line(0,-1){40}}
\put(120,68){\line(0,-1){40}}
\put(46,-15){\scriptsize{\ttfamily FOLE}}
\put(22,-35){\scriptsize{\textsf{architecture}}}
\end{picture}
\end{tabular}}}
\end{tabular}}
}}}}
\end{flushleft}
In the relational model there are two approaches for database management:
the relational algebra,
which defines an imperative language,
and
the relational calculus, which defines a declarative language.
The paper 
``Relational Operations in \texttt{FOLE}''
\cite{kent:fole:rel:ops}
represents relational algebra
by
expressing the relational operations of database theory in a clear and implementable representation.
The relational calculus
will be represented in \texttt{FOLE} in a future paper.

\subsection{Overview}\label{subsec:overview}

The first-order logical environment {\ttfamily FOLE} (Kent~\cite{kent:iccs2013}) is a framework for 
defining the semantics and formalism of logic and databases in an integrated and coherent fashion.
Institutions in general, and logical environments in particular, 
give equivalent heterogeneous and homogeneous representations for logical systems.
{\ttfamily FOLE} is an institution, 
since ``satisfaction is invariant under change of notation".
{\ttfamily FOLE} is a logical environment, 
since ``satisfaction respects structure linkage''.
As an institution,
the architecture of {\ttfamily FOLE} consists of
languages as indexing components,
structures to represent semantic content,
specifications to represent formal content, and
logics to combine formalism with semantics.
{\ttfamily FOLE} structures are interpreted as relational/logical databases.
%

%
%

In \S~\ref{sub:sec:era:data:model} and \S~\ref{sub:sec:fole:comps}
we show how the {\ttfamily ERA} data model is represented in {\ttfamily FOLE}
by connecting elements of the {\ttfamily ERA} data model to components of the {\ttfamily FOLE} structure concept.
\S~\ref{sub:sec:era:data:model} discusses the direct lower-level connection 
between the {\ttfamily ERA} elements (attributes, entities, relations) 
and the {\ttfamily FOLE} components (type domains and entity classifications).
%
\footnote{The theory of classifications and infomorphisms 
is discussed in the book 
{\itshape Information Flow} by Barwise and Seligman~\cite{barwise:seligman:97}.}
%
\S~\ref{sub:sec:fole:comps} discusses the abstract higher-level representation 
of the {\ttfamily ERA} data model within the {\ttfamily FOLE} architecture. 
%
\footnote{In a direct fashion,
we show how the {\ttfamily ERA} entity notion is represented by the {\ttfamily FOLE} entity classification,
and how the {\ttfamily ERA} attribute notion is represented by the {\ttfamily FOLE} type domain (attribute classification).
In an indirect fashion,
we show how the {\ttfamily ERA} relation notion is represented 
in general by the {\ttfamily FOLE} designation plus mixed data model, 
and in particular by the sequents in {\ttfamily FOLE} specifications 
(discussed further in Kent~\cite{kent:fole:era:supstruc}).}
In addition,
we give a rudimentary description of the interpretation of {\ttfamily FOLE} structures
in \S~\ref{sub:sub:sec:data:model:interp}.
In \S~\ref{sub:sec:connections}
we connect {\ttfamily FOLE} to Sowa's knowledge representation hierarchy (Sowa~\cite{sowa:kr})
and through linearization to the {\ttfamily Olog} data model (Spivak and Kent~\cite{spivak:kent:olog}).
%

This paper,
which is concerned with the {\ttfamily FOLE} foundation,
is illustrated in Fig.~\ref{fig:cxt:struc}
of
\S\,\ref{sub:sub:sec:data:model:struc}
and is centered on the mathematical context of structures.
%
\footnote{Following the original discussion of {\ttfamily FOLE} (Kent~\cite{kent:iccs2013}), 
we use 
``mathematical context'' 
for the mathematical term ``category'',
``passage'' for the term ``functor'', and
``bridge'' for the term ``natural transformation''.}
%
{\ttfamily FOLE} structures sit at the bottom of the
classification form
of the {\ttfamily FOLE} architecture.
%
The {\ttfamily FOLE} Superstructure
(Kent~\cite{kent:fole:era:supstruc}),
which is concern with the formalism and semantics of first-order logic, and
the {\ttfamily FOLE} interpretation,
which is concerned with database interpretation,
are presented in the two papers that follow this one.
Two further papers are pending on the integration of federated systems of knowledge:
one discusses integration over a fixed type domain and
the other discusses integration over a fixed universe.


\begin{table}
\begin{center}
{\scriptsize{\setlength{\extrarowheight}{1.6pt}
{\begin{tabular}{l@{\hspace{10pt}}l}
{\fbox{\begin{tabular}[t]{
|l@{\hspace{5pt}}l@{\hspace{2pt}:\hspace{8pt}}l|
}
\hline
\S\,\ref{sub:sec:fole}
&
Fig.~\ref{fig:fole:data:model}
&
{\ttfamily ERA} Data Model in {\ttfamily FOLE}
\\
\hline\hline
\S\,\ref{sub:sub:sec:data:model:rel}
&
Fig.~\ref{fig:eg}
&
Example
\\
\hline\hline
\S\,\ref{sub:sub:sec:data:model:struc}
&
Fig.~\ref{fig:fole:struc}
&
Structure
\\
&
Fig.~\ref{fig:interp:struc:obj}
&
Interpreted Structure
\\
&
Fig.~\ref{fig:fole:struc:mor}
&
Structure Morphism
\\
&
Fig.~\ref{fig:interp:struc:mor}
&
Interpreted Structure Morphism
\\
&
Fig.~\ref{fig:cxt:struc}
&
{\ttfamily FOLE} Foundation
\\
\hline\hline
\S\,\ref{sub:sub:sec:analogy}
&
Fig.~\ref{analogy}
&
Analogy
\\\hline
\end{tabular}}}
&
{\fbox{\begin{tabular}[t]{|l@{\hspace{5pt}}l@{\hspace{2pt}:\hspace{8pt}}l|}
\hline
\S\,\ref{sub:sec:fole}
&
Tbl.~\ref{tbl:fole:era}
&
{\ttfamily FOLE}-{\ttfamily ERA} Correspondence
\\\cline{1-2}
\S\ref{subsec:overview}
&
Tbl.\,\ref{tbl:figs:tbls}
&
Figures and Tables
\\\hline\hline
\S\,\ref{sub:sub:sec:analogy}
&
Tbl.~\ref{matrix}
&
Matrix of six central categories
\\\hline
\end{tabular}}}
\end{tabular}}
}}
\end{center}
\caption{Figures and Tables}
\label{tbl:figs:tbls}
\end{table}

\comment{
{\scriptsize{\begin{description}
\item[To Do:] \mbox{}
\begin{itemize}
\item 
Change the list map notation 
from
$\mathrmbf{List}(Y_{2})\xleftarrow[{\scriptscriptstyle\sum}_{g}]{\mathrmbf{List}(g)}\mathrmbf{List}(Y_{1})$
to
$\mathrmbf{List}(Y_{2})\xleftarrow[{\scriptstyle\exists}_{g}]{\mathrmbf{List}(g)}\mathrmbf{List}(Y_{1})$,
since this is in the semantic aspect not the formal aspect.
\item 
Define the relation passage
$\mathrmbf{Rel}(\mathcal{A}_{2})\xleftarrow{\mathrmbfit{rel}_{{\langle{f,g}\rangle}}}\mathrmbf{Rel}(\mathcal{A}_{2})$
for a typed domain morphism
$\mathcal{A}_{2}\xrightleftharpoons{{\langle{f,g}\rangle}}\mathcal{A}_{1}$.
\end{itemize}
\end{description}}}
}



\section{{\ttfamily ERA} Data Model}\label{sub:sec:era:data:model}

\subsection{Attributes.}\label{sub:sub:sec:data:model:attr}

%
In the {\ttfamily ERA} data model,
attributes are represented by a typed domain
consisting of a collection of data types.
In {\ttfamily FOLE},
a typed domain is represented by an attribute classification 
$\mathcal{A} = {\langle{X,Y,\models_{\mathcal{A}}}\rangle}$
consisting of a set of attribute types (sorts) $X$,
a set of attribute instances (data values) $Y$ and
an attribute classification relation $\models_{\mathcal{A}}{\,\subseteq\,}X{\times}Y$. 
For each sort (attribute type) $x \in X$,
the data domain of that type is the $\mathcal{A}$-extent
$\mathcal{A}_{x}=\mathrmbfit{ext}_{\mathcal{A}}(x) = \{ y \in Y \mid y \models_{\mathcal{A}} x \}$.
The passage
$X \xrightarrow{\mathrmbfit{ext}_{\mathcal{A}}} {\wp}Y$
maps a sort $x{\,\in\,}X$ to its data domain ($\mathcal{A}$-extent) $\mathcal{A}_{x}{\;\subseteq\;}Y$.

An $X$-signature (header) is a sort list ${\langle{I,s}\rangle}$, 
where $I\xrightarrow{\,s\,}X$ is a map from an indexing set (arity) 
$I$ to the set of 
sorts $X$.
A more visual representation
for this signature is 
$({\cdots\,}s_{i}{\,\cdots}{\,\mid\,}i{\,\in\,}I)$.
The mathematical context of $X$-signatures is
$\mathrmbf{List}(X)$.
\footnote{$\mathrmbf{List}(X)$ is the comma context
$\mathrmbf{List}(X) = (\mathrmbf{Set}{\,\downarrow\,}X)$
of $X$-signatures,
where an object ${\langle{I,s}\rangle}$ is an $X$-signature 
and a morphism ${\langle{I',s'}\rangle} \xrightarrow{h} {\langle{I,s}\rangle}$
is 
an arity function $I'\xrightarrow{h}I$
that preserves signatures $h{\,\cdot\,}s = s'$;
visually,
$({\cdots\,}s'_{i'}{\,\cdots}{\,\mid\,}i'{\,\in\,}I')=
({\cdots\,}s_{h(i')}{\,\cdots}{\,\mid\,}i'{\,\in\,}I')$.}
\footnote{The header for a database table is a signature 
(list of sorts)
${\langle{I,s}\rangle}{\;\in\;}\mathrmbf{List}(X)$. 
Pairs $(i : s_{i})$ from a signature ${\langle{I,s}\rangle}$ are called attributes
(see \S~\ref{sub:sub:sec:data:model:sch}).
Examples of attributes are `{\ttfamily (name : Str)}', `{\ttfamily (age : Natno)}'.} 
A $Y$-tuple (row) is an list of data values ${\langle{J,t}\rangle}$, 
where $J\xrightarrow{\,t\,}Y$ is a map 
from an indexing set (arity) $J$ 
to the set of data values 
$Y$.
A more visual representation
for this tuple is 
$({\cdots\,}t_{j}{\,\cdots}{\,\mid\,}j{\,\in\,}J)$.
The mathematical context of $Y$-tuples is
$\mathrmbf{List}(Y)$.
%
The attribute list classification 
$\mathrmbf{List}(\mathcal{A}) = {\langle{\mathrmbf{List}(X),\mathrmbf{List}(Y),\models_{\mathrmbf{List}(\mathcal{A})}}\rangle}$
has $X$-signatures as types and
$Y$-tuples as instances,
with classification by common arity and universal $\mathcal{A}$-classification:
a $Y$-tuple ${\langle{J,t}\rangle}$ 
is classified by 
an $X$-signature ${\langle{I,s}\rangle}$ 
when
$J = I$ and
$t_{k} \models_{\mathcal{A}} s_{k}$
for all $k \in J = I$.

\subsection{Entities.}\label{sub:sub:sec:data:model:ent}

We distinguish between an entity instance and an entity type. 
An entity type is a category of existence;
entity types classify entity instances.
There might be many instances of an entity type,
and an entity instance can be classified by many types.
An entity instance (entity, for short) is also called an object.
Every entity is uniquely identified by a key.
In {\ttfamily FOLE},
entities and their types are collected together locally in an entity classification 
$\mathcal{E} = {\langle{R,K,\models_{\mathcal{E}}}\rangle}$
consisting of a set of entity types $R$,
a set of entity instances (keys) $K$ and
an entity classification relation $\models_{\mathcal{E}}{\subseteq\,}R{\times}K$. 
In the database interpretation
in \S\ref{sub:sub:sec:data:model:interp},
each entity type $r \in R$ is regarded to be the name for a relation (or table) in the database:
for each entity type (relation name) $r \in R$,
the set of primary keys for that type is the $\mathcal{E}$-extent
$\mathcal{E}_{r}=\mathrmbfit{ext}_{\mathcal{E}}(r) = \{ k \in K \mid k \models_{\mathcal{E}} r \}$.

\subsection{Relations.}\label{sub:sub:sec:data:model:rel}

Here we discuss how the relational aspect of the {\ttfamily ERA} data model is handled in {\ttfamily FOLE}.
Some relations are special.
One example is subtyping,
which specifies that one category of existence is more general than another.
This arises when representing the taxonomic aspect of ontologies.
Subtyping is handled by the binary sequents
\footnote{A sequent $\varphi{\;\vdash\;}\psi$ expresses 
interpretation widening
between formulas.}
in {\ttfamily FOLE} specifications
(discussed further in the 
 (Kent~\cite{kent:fole:era:supstruc})). 
Some many-to-one relationships can be represented as attributes.
But in general,
many-to-many relationships are represented in {\ttfamily FOLE} as entities,
whose attributes,
each of which plays a thematic role for the relationship,
may be other entities.
\footnote{As an example,
the ``marriage'' binary relation
can be represented as a {\itshape{Marriage}} entity
with {\itshape{wife}} and {\itshape{husband}} attributes that are themselves {\itshape{Person}} entities.}
%


Consider the example (Fig.~\ref{fig:eg}) of a simple entity-relationship-attribute diagram.
Here we have 
three entities (represented by rectangles), 
two relationships (represented by diamonds) and 
numerous attributes (represented by ovals).
The {\itshape{works\_on}} relationship is many-to-many,
and so we can represent this in {\ttfamily FOLE} as an entity
type {\itshape{Activity}}
with four attributes:
{\itshape{entry\_date}} of sort {\itshape{Date}},
{\itshape{job\_descr}} of sort {\itshape{String}},
{\itshape{employee}} of sort {\itshape{Employee}}, and
{\itshape{project}} of sort {\itshape{Project}}.
Note that attributes {\itshape{employee}} and {\itshape{project}} 
are foreign keys of the {\itshape{Activity}} entity.
\footnote{The {\itshape{Employee}} type,
which plays the {\itshape{employee}} thematic role for the {\itshape{works\_on}} relationship,
is both an entity type and an attribute type (sort);
any value in the {\itshape{Employee}} data domain
is a key of 
the {\itshape{Employee}} entity
and a foreign key
of the {\itshape{Activity}} entity.}
Since the {\itshape{works\_for}} relationship is many-to-one
without any attributes of its own,
we can represent this as an attribute called {\itshape{dept}}
of sort {\itshape{Department}}.
This is a foreign key
of the {\itshape{Employee}} entity.
%
\begin{figure}
\begin{center}
{{\begin{tabular}{c}
\setlength{\unitlength}{0.7pt}
\begin{picture}(360,250)(-45,15)
\put(160,250){\begin{picture}(50,18)(25,9)
\put(25,9){\makebox(0,0){\scriptsize{$\text{\itshape{name}}$}}}
\put(25,9){\oval(50,18)}
\put(25,0){\line(1,-1){26}}
\end{picture}}
\put(240,250){\begin{picture}(50,18)(25,9)
\put(25,9){\makebox(0,0){\scriptsize{$\text{\itshape{id\_num}}$}}}
\put(25,9){\oval(50,18)}
\put(25,0){\line(-1,-1){26}}
\end{picture}}
\put(200,200){\begin{picture}(60,30)(30,15)
\put(30,15){\makebox(0,0){\scriptsize{$\text{\itshape{Employee}}$}}}
\put(0,30){\line(1,0){60}}
\put(0,0){\line(1,0){60}}
\put(0,0){\line(0,1){30}}
\put(60,0){\line(0,1){30}}
\end{picture}}
\put(40,200){\begin{picture}(60,40)(30,20)
\put(30,21){\makebox(0,0){\scriptsize{$\text{\itshape{works\_for}}$}}}
\put(0,20){\line(3,2){30}}
\put(0,20){\line(3,-2){30}}
\put(60,20){\line(-3,2){30}}
\put(60,20){\line(-3,-2){30}}
\put(60,20){\line(1,0){100}}
\put(30,0){\line(0,-1){25}}
\put(154,25){\makebox(0,0){\scriptsize{$\text{\itshape{n}}$}}}
\put(35,-19){\makebox(0,0){\scriptsize{$\text{\itshape{1}}$}}}
\end{picture}}
\put(40,140){\begin{picture}(60,30)(30,15)
\put(30,15){\makebox(0,0){\scriptsize{$\text{\itshape{Department}}$}}}
\put(-3,30){\line(1,0){66}}
\put(-3,0){\line(1,0){66}}
\put(-3,0){\line(0,1){30}}
\put(63,0){\line(0,1){30}}
\put(-30,-35){\begin{picture}(50,18)(25,9)
\put(25,9){\makebox(0,0){\scriptsize{$\text{\itshape{name}}$}}}
\put(25,9){\oval(50,18)}
\put(25,18){\line(2,1){51}}
\end{picture}}
\put(30,-35){\begin{picture}(50,18)(25,9)
\put(25,9){\makebox(0,0){\scriptsize{$\text{\itshape{id\_num}}$}}}
\put(25,9){\oval(50,18)}
\put(25,18){\line(0,1){26}}
\end{picture}}
\put(90,-35){\begin{picture}(50,18)(25,9)
\put(25,9){\makebox(0,0){\scriptsize{$\text{\itshape{location}}$}}}
\put(25,9){\oval(50,18)}
\put(25,18){\line(-2,1){51}}
\end{picture}}
\end{picture}}
\put(200,140){\begin{picture}(60,40)(30,20)
\put(30,21){\makebox(0,0){\scriptsize{$\text{\itshape{works\_on}}$}}}
\put(0,20){\line(3,2){30}}
\put(0,20){\line(3,-2){30}}
\put(60,20){\line(-3,2){30}}
\put(60,20){\line(-3,-2){30}}
\put(30,40){\line(0,1){25}}
\put(30,0){\line(0,-1){25}}
\put(23,60){\makebox(0,0){\scriptsize{$\text{\itshape{m}}$}}}
\put(23,-20){\makebox(0,0){\scriptsize{$\text{\itshape{n}}$}}}
\put(100,40){\begin{picture}(60,18)(30,9)
\put(30,9){\makebox(0,0){\scriptsize{$\text{\itshape{entry\_date}}$}}}
\put(30,9){\oval(60,18)}
\put(0,9){\line(-1,-1){16}}
\end{picture}}
\put(100,0){\begin{picture}(60,18)(30,9)
\put(30,9){\makebox(0,0){\scriptsize{$\text{\itshape{job\_descr}}$}}}
\put(30,9){\oval(60,18)}
\put(0,9){\line(-1,1){16}}
\end{picture}}
\end{picture}}
\put(200,80){\begin{picture}(60,30)(30,15)
\put(30,15){\makebox(0,0){\scriptsize{$\text{\itshape{Project}}$}}}
\put(0,30){\line(1,0){60}}
\put(0,0){\line(1,0){60}}
\put(0,0){\line(0,1){30}}
\put(60,0){\line(0,1){30}}
\put(-30,-35){\begin{picture}(50,18)(25,9)
\put(25,9){\makebox(0,0){\scriptsize{$\text{\itshape{name}}$}}}
\put(25,9){\oval(50,18)}
\put(25,18){\line(2,1){51}}
\end{picture}}
\put(30,-35){\begin{picture}(50,18)(25,9)
\put(25,9){\makebox(0,0){\scriptsize{$\text{\itshape{id\_num}}$}}}
\put(25,9){\oval(50,18)}
\put(25,18){\line(0,1){26}}
\end{picture}}
\put(90,-35){\begin{picture}(50,18)(25,9)
\put(25,9){\makebox(0,0){\scriptsize{$\text{\itshape{budget}}$}}}
\put(25,9){\oval(50,18)}
\put(25,18){\line(-2,1){51}}
\end{picture}}
\end{picture}}
\end{picture}
\end{tabular}}}
\end{center}
\caption{Example}
\label{fig:eg}
\end{figure}
%

\newpage
\section{{\ttfamily FOLE} Components}\label{sub:sec:fole:comps}

\subsection{Schema.}\label{sub:sub:sec:data:model:sch}

The type aspect of the {\ttfamily ERA} data model is gathered together into a schema.
A \emph{schema} $\mathcal{S}={\langle{R,\sigma,X}\rangle}$
consists of
a set of sorts (attribute types) $X$, 
a set of entity types $R$ and
a signature map $R \xrightarrow{\sigma} \mathrmbf{List}(X)$.
Within the schema $\mathcal{S}$,
we think of each $r \in R$ as being 
an entity type 
that is locally described by the associated $X$-signature 
$\sigma(r) = {\langle{I,s}\rangle} \in \mathrmbf{List}(X)$.
\footnote{There is an associated arity function
$R\xrightarrow[\sigma{\,\circ\,}\mathrmbfit{set}]{\alpha}\mathrmbf{Set}
:r\xmapsto{\sigma}{\langle{I,s}\rangle}\xmapsto{\mathrmbfit{set}}I$.}
A more visual representation
for this signature mapping is 
$r\xmapsto{\sigma}({\cdots\,}s_{i}{\,\cdots}{\,\mid\,}i{\,\in\,}I)$.
An {\ttfamily ERA}-style visualization might be
\begin{tabular}{c}
\begin{picture}(60,10)(0,2)
\thicklines
\put(0,0){\framebox(18,10){\scriptsize{$r$}}}
\put(30,8){\makebox(0,0){\footnotesize{$\xrightarrow{\;\;i\;\;}$}}}
\put(50,5){\oval(20,10)}\put(50,5){\makebox(0,0){\scriptsize{$s_{i}$}}}
\end{picture}
\end{tabular},
where 
the box encloses the entity type $r{\,\in\,}R$, 
the oval encloses the attribute type $s_{i}{\,\in\,}X$,
and the arrow is labeled with the index $i{\,\in\,}I$.
For example,
\begin{tabular}{c}
\begin{picture}(85,10)(0,2)
\thicklines
\put(0,0){\framebox(25,10){\scriptsize{$\mathtt{Person}$}}}
\put(40,7){\makebox(0,0){\footnotesize{$\xrightarrow{\;\mathtt{name}\;}$}}}
\put(70,5){\oval(28,10)}\put(70,5){\makebox(0,0){\scriptsize{$\mathtt{String}$}}}
\end{picture}
\end{tabular} or
\begin{tabular}{c}
\begin{picture}(105,10)(0,2)
\thicklines
\put(0,0){\framebox(35,10){\scriptsize{$\mathtt{Employee}$}}}
\put(50,7.2){\makebox(0,0){\footnotesize{$\xrightarrow{\;\mathtt{dept}\;}$}}}
\put(84,5.2){\oval(41,10)}\put(84,5){\makebox(0,0){\scriptsize{$\mathtt{Department}$}}}
\end{picture}
\end{tabular}.

The entity type $r$ in the {\ttfamily ERA} data model 
corresponds to
the 
relation symbol $r$ in {\ttfamily FOLE}/{\ttfamily MSFOL}.
Either representation is a kind of nexus. 
A schema corresponds to a multi-sorted first-order logical language in
the {\ttfamily FOLE}/{\ttfamily MSFOL} approach to knowledge representation.
%
\footnote{Formulas based on relation symbols can be inductively defined, 
thus forming extended schemas 
(Kent~\cite{kent:db:sem}).
Terms composed of function symbols can be added as constraints between formulas.}
In the database interpretation of {\ttfamily FOLE}
(Kent~\cite{kent:fole:era:db}),
we think of $r$ as being a relation name with associated header $\sigma(r) = {\langle{I,s}\rangle}$.
%
\begin{sloppypar}
We formally link schemas with morphisms.
A \emph{schema morphism}  
$
\mathcal{S}_{2} = {\langle{R_{2},\sigma_{2},X_{2}}\rangle} \stackrel{{\langle{r,f}\rangle}}{\Longrightarrow}
{\langle{R_{1},\sigma_{1},X_{1}}\rangle} = \mathcal{S}_{1}
$
from schema $\mathcal{S}_{2}$ to schema $\mathcal{S}_{1}$
consists of 
an sort 
function $f : X_{2} \rightarrow X_{1}$
and an entity type function $r : R_{2} \rightarrow R_{1}$, 
which preserve signatures 
by satisfying the condition
$r \cdot \sigma_{1} = \sigma_{2} \cdot {\scriptstyle\sum}_{f}$.
\end{sloppypar}
\begin{center}
\begin{tabular}{c}
\setlength{\unitlength}{0.5pt}
\begin{picture}(290,110)(10,-20)
\put(0,80){\makebox(0,0){\footnotesize{$r_{2} \in R_{2}$}}}
\put(60,0){\makebox(0,0)[r]{\footnotesize{${\langle{I,s}\rangle}=({\cdots\,}s_{i}{\,\cdots}{\,\mid\;}i{\,\in\,}I)$}}}
\put(0,-22){\makebox(0,0){\footnotesize{$s_{i}{\,\in\,}X_{2}$}}}
\put(260,80){\makebox(0,0){\footnotesize{$r(r_{2}) \in R_{1}$}}}
\put(188,0){\makebox(0,0)[l]{\footnotesize{
$({\cdots\,}f(s_{i}){\,\cdots}{\,\mid\;}i{\,\in\,}I)={\sum}_{f}(I,s)$}}}
\put(260,-22){\makebox(0,0){\footnotesize{$f(s_{i}){\,\in\,}X_{1}$}}}
\put(0.2,49){\makebox(0,0){\scriptsize{$-$}}}
\put(0,40){\makebox(0,0){\normalsize{$\downarrow$}}}
\put(-10,40){\makebox(0,0)[r]{\scriptsize{$\sigma_{2}$}}}
\put(260.2,49){\makebox(0,0){\scriptsize{$-$}}}
\put(260,40){\makebox(0,0){\normalsize{$\downarrow$}}}
\put(270,40){\makebox(0,0)[l]{\scriptsize{$\sigma_{1}$}}}
\put(130,80){\makebox(0,0){\normalsize{$\xmapsto{\;\;r\;\;}$}}}
\put(130,10){\makebox(0,0){\normalsize{$\xmapsto{{\scriptstyle\sum}_{f}}$}}}
\put(-140,40){\makebox(0,0)[r]{\footnotesize{$\mathcal{S}_{2}\left\{\rule{0pt}{33pt}\right.$}}}
\put(445,40){\makebox(0,0)[l]{\footnotesize{$\left.\rule{0pt}{33pt}\right\}\mathcal{S}_{2}$}}}
\end{picture}
\end{tabular}
\end{center}
Let $\mathrmbf{Sch}$ denote the mathematical context of schemas and their morphisms.

\newpage
\subsection{Universe.}\label{sub:sub:sec:data:model:univ}

The instance aspect of the {\ttfamily ERA} data model is gathered together into a universe.
A \emph{universe} $\mathcal{U}={\langle{K,\tau,Y}\rangle}$
consists of
a set of values (attribute instances) $Y$, 
a set of keys (entity instances) $K$ and
a tuple map $K \xrightarrow{\tau} \mathrmbf{List}(Y)$.
Within the universe $\mathcal{U}$,
we think of each key $k \in K$ as being 
an identifier or name for an object
that is locally described by the associated tuple of values 
$\tau(k) = {\langle{J,t}\rangle}\in\mathrmbf{List}(Y)$.
A more visual representation
for this tuple mapping is 
$k\xmapsto{\tau}({\cdots\,}t_{j}{\,\cdots}{\,\mid\,}j{\,\in\,}J)$.
Note that, no typing has been mentioned here and no typing restrictions are required.
In a universe by itself,
we do not require the data values $t_{j}$ to be members of any special data-types.

An element of a universe $\mathcal{U}={\langle{K,\tau,Y}\rangle}\in\mathrmbf{Univ}$
is a key $\bigl(k\xmapsto{\tau}{\langle{J,t}\rangle}\bigr)$ 
with associated list.
We can think of such universe elements
as object descriptions without attached typing 
\underline{or}
as tuples untethered from a database table.
They develop meaning by being classified by schema elements 
$\bigl(r\xmapsto{\sigma}{\langle{I,s}\rangle}\bigr)$ 
in a structure
(\S~\ref{sub:sub:sec:data:model:struc}).
\footnote{They are somewhat like genes (bits of DNA) without the genomic structure that provides interpretation.}
Hence,
a {\ttfamily FOLE} universe
is like the key$-$value${\cdot}$list store at the heart of Google's Spanner database (Google~\cite{google:spanner}):
{\itshape\footnotesize{\begin{quotation}
\noindent
``Spanner's data model is not purely relational, 
in that rows must have names. 
More precisely, every table is required to have an ordered set of one or more primary-key columns. 
This requirement is where Spanner still looks like a key-value store: 
the primary keys form the name for a row, and 
each table defines a mapping from the primary-key columns to the non-primary-key columns.''
\footnote{When the universe $\mathcal{U}$ is the instance aspect of a {\ttfamily FOLE} structure $\mathcal{M}$
with typed domain $\mathcal{A}$,
in the database interpretation of that structure
(\S~\ref{sub:sub:sec:data:model:interp}
and
Kent~\cite{kent:fole:era:tbl}),
we think of the entity instance $k$ as being a primary key
that indexes a row $\tau(k) = {\langle{I,t}\rangle}$ in the table associated 
with the relation symbol $r \in R$ with associated header $\sigma(r) = {\langle{I,s}\rangle}$.
A more visual representation
for this tuple mapping is 
$k\xmapsto{\tau}({\cdots\,}t_{i}{\,\cdots}{\,\mid\,}i{\,\in\,}I,t_{i}{\,\in\,}\mathcal{A}_{s_{i}})$,
where $\mathcal{A}_{s_{i}}$ is the data-type for sort $s_{i}{\,\in\,}X$.
Here,
we do require the data values $t_{i}$ to be members of the special data-types $\mathcal{A}_{s_{i}}$.}
\end{quotation}}}

We semantically link universes with morphisms.
A \emph{universe morphism}  
$\mathcal{U}_{2} = {\langle{K_{2},\tau_{2},Y_{2}}\rangle} 
\xLeftarrow{{\langle{k,g}\rangle}}
{\langle{K_{1},\tau_{1},Y_{1}}\rangle} = \mathcal{U}_{1}$
consists of 
a value (attribute instance) function $Y_{2} \xleftarrow{g} Y_{1}$
and a key (entity instance) function $K_{2} \xleftarrow{k} K_{1}$, 
which preserve tuples (instance lists) by satisfying the condition
$k \cdot \tau_{2} = \tau_{1} \cdot {\scriptstyle\sum}_{g}$.
\begin{center}
\begin{tabular}{c}
\setlength{\unitlength}{0.5pt}
\begin{picture}(290,125)(-45,-30)
\put(0,80){\makebox(0,0){\footnotesize{$k(k_{1}) \in K_{2}$}}}
\put(60,0){\makebox(0,0)[r]{\footnotesize{${\sum}_{g}(J,t)=({\cdots\,}g(t_{j}){\,\cdots}{\,\mid\;}j{\,\in\,}J)$}}}
\put(0,-22){\makebox(0,0){\footnotesize{$g(t_{j}){\,\in\,}Y_{2}$}}}
\put(260,80){\makebox(0,0){\footnotesize{$k_{1} \in K_{1}$}}}
\put(188,0){\makebox(0,0)[l]{\footnotesize{
$({\cdots\,}t_{j}{\,\cdots}{\,\mid\;}j{\,\in\,}J)={\langle{J,t}\rangle}$}}}
\put(260,-22){\makebox(0,0){\footnotesize{$t_{j}{\,\in\,}Y_{1}$}}}
\put(0.2,49){\makebox(0,0){\scriptsize{$-$}}}
\put(0,40){\makebox(0,0){\normalsize{$\downarrow$}}}
\put(-10,40){\makebox(0,0)[r]{\scriptsize{$\tau_{2}$}}}
\put(260.2,49){\makebox(0,0){\scriptsize{$-$}}}
\put(260,40){\makebox(0,0){\normalsize{$\downarrow$}}}
\put(270,40){\makebox(0,0)[l]{\scriptsize{$\tau_{1}$}}}
\put(130,80){\makebox(0,0){\normalsize{$\stackrel{k}{\longmapsfrom}$}}}
\put(130,10){\makebox(0,0){\normalsize{$\stackrel{{\scriptstyle\sum}_{g}}{\longmapsfrom}$}}}
\put(-188,40){\makebox(0,0)[r]{\footnotesize{$\mathcal{U}_{2}\left\{\rule{0pt}{33pt}\right.$}}}
\put(396,40){\makebox(0,0)[l]{\footnotesize{$\left.\rule{0pt}{33pt}\right\}\mathcal{U}_{1}$}}}
\end{picture}
\end{tabular}
\end{center}
Let $\mathrmbf{Univ}$ denote the mathematical context of universes and their morphisms.


\newpage
\subsection{Structure.}\label{sub:sub:sec:data:model:struc}

The complete {\ttfamily ERA} data model is incorporated into the notion of
a (model-theoretic) structure in the {\ttfamily FOLE} representation of knowledge. 

%
\paragraph{Structures.}
%
A {\ttfamily FOLE} \emph{structure}
$\mathcal{M} = {\langle{\mathcal{E},\sigma,\tau,\mathcal{A}}\rangle}$
is a hypergraph of classifications 
(Fig.~\ref{fig:fole:struc})
--- a two-dimensional construct with the following components:
\begin{center}
{\footnotesize{$\begin{array}{r@{\hspace{8pt}}r@{\hspace{8pt}}r@{\hspace{8pt}=\hspace{8pt}}l@{\hspace{8pt}=\hspace{8pt}}l}
\text{attribute classification:}
&\text{\emph{typed domain}}
& \mathcal{A}
& \mathrmbfit{attr}(\mathcal{M}) 
& {\langle{X,Y,\models_{\mathcal{A}}}\rangle}
\\
\text{entity classification:}
&
& \mathcal{E}
& \mathrmbfit{ent}(\mathcal{M}) 
& {\langle{R,K,\models_{\mathcal{E}}}\rangle}
\\
\text{type hypergraph:}
& \text{\emph{schema}}
& \mathcal{S}
& \mathrmbfit{sch}(\mathcal{M})
& {\langle{R,\sigma,X}\rangle}
\\
\text{instance hypergraph:}
&\text{\emph{universe}}
& \mathcal{U}
& \mathrmbfit{univ}(\mathcal{M}) 
& {\langle{K,\tau,Y}\rangle}
\end{array}$}}
\end{center}
and 
a list designation ${\langle{\sigma,\tau}\rangle} : \mathcal{E} \rightrightarrows \mathrmbf{List}(\mathcal{A})$
with signature map $R \xrightarrow{\sigma} \mathrmbf{List}(X)$
and tuple map $K \xrightarrow{\tau} \mathrmbf{List}(Y)$,
whose defining condition states that:
if entity $k{\,\in\,}K$ is of type $r{\,\in\,}R$,
then the description tuple $\tau(k)={\langle{J,t}\rangle}$ 
is the same ``size'' ($J=I$) as
the signature $\sigma(r)={\langle{I,s}\rangle}$ 
and each data value $t_{n}$ is of sort $s_{n}$;
or interpretively
(\S~\ref{sub:sub:sec:data:model:interp}
and
Kent~\cite{kent:fole:era:tbl}),
in a database table 
all rows are classified by the table header.
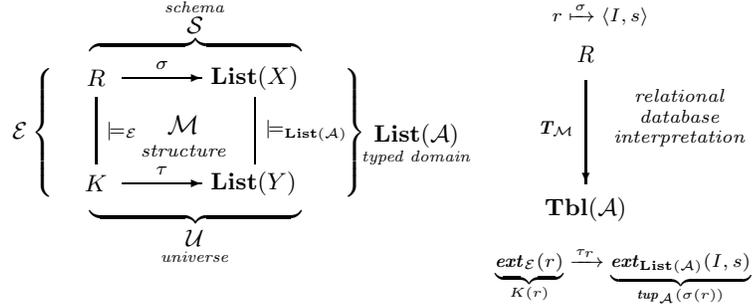
\begin{figure}
\begin{center}
{{\begin{tabular}{@{\hspace{80pt}}c@{\hspace{85pt}}c}
{{\begin{tabular}{c}
\begin{tabular}{c}
\setlength{\unitlength}{0.5pt}
\begin{picture}(140,120)(-10,-10)
\put(0,80){\makebox(0,0){\footnotesize{$R$}}}
\put(0,0){\makebox(0,0){\footnotesize{$K$}}}
\put(87,80){\makebox(0,0)[l]{\footnotesize{$\mathrmbf{List}(X)$}}}
\put(87,0){\makebox(0,0)[l]{\footnotesize{$\mathrmbf{List}(Y)$}}}
\put(8,40){\makebox(0,0)[l]{\scriptsize{$\models_{\mathcal{E}}$}}}
\put(128,40){\makebox(0,0)[l]{\scriptsize{$\models_{\mathrmbf{List}(\mathcal{A})}$}}}
\put(50,90){\makebox(0,0){\scriptsize{$\sigma$}}}
\put(50,10){\makebox(0,0){\scriptsize{$\tau$}}}
\put(20,80){\vector(1,0){60}}
\put(20,0){\vector(1,0){60}}
\put(0,65){\line(0,-1){50}}
\put(120,65){\line(0,-1){50}}
\put(66,45){\makebox(0,0){\normalsize{$\mathcal{M}$}}}
\put(66,25){\makebox(0,0){{\scriptsize{\itshape{structure}}}}}
\put(75,120){\makebox(0,0){\footnotesize{$\overset{\textstyle{
\stackrel{\text{\emph{schema}}}{\mathcal{S}\rule{0pt}{1pt}}}}{\overbrace{\rule{80pt}{0pt}}}$}}}
\put(75,-40){\makebox(0,0){\footnotesize{$\underset{\textstyle{
\stackrel{\textstyle{\mathcal{U}}}{\scriptstyle{\text{\emph{universe}}}}}}{\underbrace{\rule{80pt}{0pt}}}$}}}
\put(-30,40){\makebox(0,0)[r]{\footnotesize{$\mathcal{E}\left\{\rule{0pt}{28pt}\right.$}}}
\put(185,40){\makebox(0,0)[l]{\footnotesize{$\left.\rule{0pt}{28pt}\right\}$}}}
\put(201,31){\makebox(0,0)[l]{\footnotesize{$
\stackrel{\textstyle{\mathrmbf{List}(\mathcal{A})}}{\scriptstyle{\text{\emph{typed domain}}}}$}}}
\end{picture}
\end{tabular}
\\\\
\end{tabular}}}
&
{{\begin{tabular}{c}
\setlength{\unitlength}{0.7pt}
\begin{picture}(80,180)(0,-10)
\put(40,122){\makebox(0,0){\footnotesize{$\mathrmit{R}$}}}
\put(40,38){\makebox(0,0){\footnotesize{$\mathrmbf{Tbl}(\mathcal{A})$}}}
\put(34,82){\makebox(0,0)[r]{\scriptsize{$\mathrmbfit{T}_{\mathcal{M}}$}}}
\put(40,108){\vector(0,-1){56}}
\put(48,144){\makebox(0,0){\scriptsize{$r\xmapsto{\sigma}{\langle{I,s}\rangle}$}}}
\put(60,2){\makebox(0,0){\scriptsize{$\underset{K(r)}{\underbrace{\mathrmbfit{ext}_{\mathcal{E}}(r)}}\xrightarrow{\tau_{r}}
\underset{\mathrmbfit{tup}_{\mathcal{A}}(\sigma(r))}{\underbrace{\mathrmbfit{ext}_{\mathrmbf{List}(\mathcal{A})}(I,s)}}$}}}
\put(55,74){\shortstack{{\scriptsize{\itshape{relational}}}\\{\scriptsize{\itshape{database}}}\\{\scriptsize{\itshape{interpretation}}}}}
\end{picture}
\end{tabular}}}
%
%
\end{tabular}}}
\end{center}
\caption{Structure}
\label{fig:fole:struc}
\end{figure}
A {\ttfamily FOLE} structure embodies the idea of an {\ttfamily ERA} data model
(compare Fig.~\ref{fig:fole:struc} with Fig.~\ref{fig:fole:data:model}).
Each community of discourse 
that incorporates the {\ttfamily ERA} data model
will have its own local {\ttfamily FOLE} structure.
\footnote{In anticipation of the discussion in \S~\ref{sub:sub:sec:data:model:interp},
we illustrate the associated tabular interpretation (\ref{eqn:tbl:ent:typ}) on the right side
of Fig.~\ref{fig:fole:struc}.}

The entity and attribute-list classifications $\mathcal{E}$ and $\mathrmbf{List}(\mathcal{A})$
are equivalent 
\footnote{Any classification $\mathcal{A}={\langle{X,Y,\models_{\mathcal{A}}}\rangle}$
is equivalent to its extent map $X\xrightarrow{\mathrmbfit{ext}_{\mathcal{A}}}{\wp}Y$.}
to their extent diagrams
$R{\;\xrightarrow{\mathrmbfit{ext}_{\mathcal{E}}}\;}\mathrmbf{Set}$ and
$\mathrmbf{List}(X){\;\xrightarrow[\mathrmbfit{tup}_{\mathcal{A}}]{\mathrmbfit{ext}_{\mathrmbf{List}(\mathcal{A})}}\;}\mathrmbf{Set}$,
and the list designation 
is equivalent to its extent diagram morphism 
$\mathrmbfit{ext}_{{\langle{\sigma,\tau}\rangle}} = {\langle{\sigma,\tau_{{\scriptscriptstyle{\llcorner}}}}\rangle}
: {\langle{R,\mathrmbfit{ext}_{\mathcal{E}}}\rangle}
\rightarrow
{\langle{\mathrmbf{List}(X),\mathrmbfit{ext}_{\mathrmbf{List}(\mathcal{A})}}\rangle}$
consisting of
the signature map $R\xrightarrow{\sigma}\mathrmbf{List}(X)$
and the bridge
$\mathrmbfit{ext}_{\mathcal{E}}
\xRightarrow{\tau_{{\scriptscriptstyle{\llcorner}}}}\sigma{\;\circ\;}\mathrmbfit{ext}_{\mathrmbf{List}(\mathcal{A})}$,
whose $r^{\text{th}}$-component
is the tuple function
$K(r)=\mathrmbfit{ext}_{\mathcal{E}}(r)
\xrightarrow{\tau_{r}}\mathrmbfit{ext}_{\mathrmbf{List}(\mathcal{A})}(\sigma(r))=\mathrmbfit{tup}_{\mathcal{A}}(\sigma(r))$
(see the tabular interpretation (\ref{eqn:tbl:ent:typ}) in \S~\ref{sub:sub:sec:data:model:interp}).
Hence,
any structure $\mathcal{M}$ has the interpretive presentation in Fig.~\ref{fig:interp:struc:obj}
(see the discussion on linearization \S~\ref{sub:sub:sec:misc:linear}).
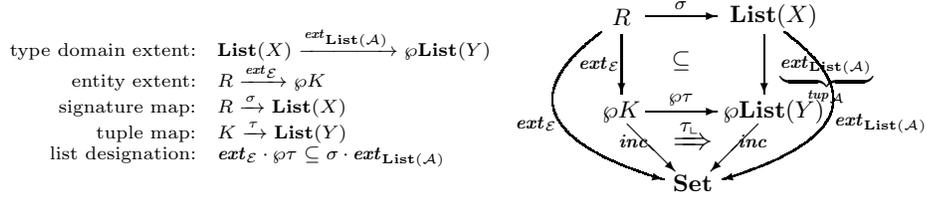
\begin{figure}
\begin{center}
{{\begin{tabular}{@{\hspace{-30pt}}c@{\hspace{30pt}}c}
{{\scriptsize{$\begin{array}[b]{r@{\hspace{8pt}}l}
\text{type domain extent:}
& 
\mathrmbf{List}(X)\xrightarrow{\mathrmbfit{ext}_{\mathrmbf{List}(\mathcal{A})}}{\wp}\mathrmbf{List}(Y)
\\
\text{entity extent:}
& 
R\xrightarrow{\mathrmbfit{ext}_{\mathcal{E}}}{\wp}K
\\
\text{signature map:}
& 
R\xrightarrow{\sigma}\mathrmbf{List}(X)
\\
\text{tuple map:}
& 
K\xrightarrow{\tau}\mathrmbf{List}(Y)
\\
\text{list designation:}
& 
\mathrmbfit{ext}_{\mathcal{E}}{\,\cdot\,}{\wp}\tau{\;\subseteq\;}\sigma{\,\cdot\,}\mathrmbfit{ext}_{\mathrmbf{List}(\mathcal{A})}
\end{array}$}}}
&
{{\begin{tabular}[b]{c}
\setlength{\unitlength}{0.45pt}
\begin{picture}(180,120)(-30,-40)
\put(0,80){\makebox(0,0){\footnotesize{$R$}}}
\put(128,80){\makebox(0,0){\footnotesize{$\mathrmbf{List}(X)$}}}
\put(0,0){\makebox(0,0){\footnotesize{${\wp}K$}}}
\put(128,0){\makebox(0,0){\footnotesize{${\wp}\mathrmbf{List}(Y)$}}}
\put(60,-60){\makebox(0,0){\footnotesize{$\mathrmbf{Set}$}}}
\put(-1,40){\makebox(0,0)[r]{\scriptsize{$\mathrmbfit{ext}_{\mathcal{E}}$}}}
\put(133,26){\makebox(0,0)[l]{\scriptsize{$
\underset{\mathrmbfit{tup}_{\mathcal{A}}}{\underbrace{\mathrmbfit{ext}_{\mathrmbf{List}(\mathcal{A})}}}$}}}
\put(-54,-10){\makebox(0,0)[r]{\scriptsize{$\mathrmbfit{ext}_{\mathcal{E}}$}}}
\put(179,-10){\makebox(0,0)[l]{\scriptsize{$\mathrmbfit{ext}_{\mathrmbf{List}(\mathcal{A})}$}}}
\put(50,90){\makebox(0,0){\scriptsize{$\sigma$}}}
\put(50,10){\makebox(0,0){\scriptsize{${\wp}\tau$}}}
\put(50,43){\makebox(0,0){{$\subseteq$}}}
\put(60,-19){\makebox(0,0){\large{$\xRightarrow{\tau_{{\scriptscriptstyle{\llcorner}}}\;}$}}}
\put(10,-25){\makebox(0,0){\scriptsize{$\mathrmbfit{inc}$}}}
\put(110,-25){\makebox(0,0){\scriptsize{$\mathrmbfit{inc}$}}}
\put(20,80){\vector(1,0){60}}
\put(20,0){\vector(1,0){60}}
\put(5,-10){\vector(1,-1){40}}
\put(115,-10){\vector(-1,-1){40}}
\put(0,65){\vector(0,-1){50}}
\put(120,65){\vector(0,-1){50}}
\qbezier(-17,67)(-120,-10)(35,-57)\put(35,-57){\vector(4,-1){0}}
\qbezier(137,67)(240,-10)(85,-57)\put(85,-57){\vector(-4,-1){0}}
\end{picture}
\end{tabular}}}
\end{tabular}}}
\end{center}
\caption{Interpreted Structure}
\label{fig:interp:struc:obj}
\end{figure}
\mbox{}\newline
In the concept of a {\ttfamily FOLE} structure 
we have abstracted the (primary) keys from the tuples that they described.
The key-embedding construction replaces keys into their tuples.
\begin{definition}(key embedding)\label{def:key:embed}
Any {\ttfamily FOLE} structure
$\mathcal{M} = {\langle{\mathcal{E},\sigma,\tau,\mathcal{A}}\rangle}$
with signature map 
$R \xrightarrow{\sigma} \mathrmbf{List}(X) : r \mapsto {\langle{I,s}\rangle}$ 
and tuple map 
$K \xrightarrow{\tau} \mathrmbf{List}(Y) : k \mapsto {\langle{I,t}\rangle}$
has a companion \emph{key embedding} structure 
$\dot{\mathcal{M}} = {\langle{\mathcal{E},\dot{\sigma},\dot{\tau},\mathcal{E}{+}\mathcal{A}}\rangle}$
consisting of
entity classification $\mathcal{E}$,
parallel sum typed domain $\mathcal{E}{+}\mathcal{A} = {\langle{R{+}X,K{+}\,Y,\models_{\mathcal{E}{+}\mathcal{A}}}\rangle}$,
\footnote{We can think of the entity classification
$\mathcal{E} = {\langle{R,K,\models_{\mathcal{E}}}\rangle}$
as a type domain.
For each sort (attribute type) $r \in R$,
the data domain of that type is the $\mathcal{E}$-extent
$\mathcal{E}_{r}=\mathrmbfit{ext}_{\mathcal{E}}(r) = \{ k \in K \mid k \models_{\mathcal{E}} r \}$.
The passage
$R \xrightarrow{\mathrmbfit{ext}_{\mathcal{E}}} {\wp}K$
maps a sort $r{\,\in\,}R$ to its data domain ($\mathcal{E}$-extent) $\mathcal{E}_{r}{\;\subseteq\;}K$.}
schema ${\langle{R,\dot{\sigma},R{+}X}\rangle}$ with
signature map
$R\xrightarrow{\dot{\sigma}}\mathrmbf{List}(R{+}X) : r\mapsto{\langle{1,r}\rangle}{+}{\langle{I,s}\rangle}$,
and 
universe ${\langle{K,\dot{\tau},K{+}\,Y}\rangle}$ with
tuple map
$K\xrightarrow{\dot{\tau}}\mathrmbf{List}(K{+}\,Y) : k\mapsto{\langle{1,k}\rangle}{+}{\langle{I,t}\rangle}$.
The signature and tuple maps are injective.
\end{definition}
%

%
\newpage
\paragraph{Structure Morphisms.}
%

In order to allow communities of discourse to interoperate,
we define the notion of a morphism between two structures 
that respects the {\ttfamily ERA} data model.
A structure morphism
$\mathcal{M}_{2}\xrightleftharpoons{{\langle{r,k,f,g}\rangle}}\mathcal{M}_{1}$
(Fig.~\ref{fig:fole:struc:mor})
from source structure 
$\mathcal{M}_{2}={\langle{\mathcal{E}_{2},{\langle{\sigma_{2},\tau_{2}}\rangle},\mathcal{A}_{2}}\rangle}$
to target structure
$\mathcal{M}_{1}={\langle{\mathcal{E}_{1},{\langle{\sigma_{1},\tau_{1}}\rangle},\mathcal{A}_{1}}\rangle}$
is defined in terms of 
the hypergraph and classification morphisms between the source and target structure components (projections):
\begin{center}
{\footnotesize{$\begin{array}{
r@{\hspace{8pt}}
r@{\hspace{4pt}=\hspace{4pt}}
r@{\hspace{4pt}}
l@{\hspace{4pt}}
l@{\hspace{4pt}=\hspace{4pt}}
l}
\text{\emph{typed domain morphism}}
& \mathcal{A}_{2}
& \mathrmbfit{attr}(\mathcal{M}_{2})
& \xrightleftharpoons{{\langle{f,g}\rangle}}
& \mathrmbfit{attr}(\mathcal{M}_{1})
& \mathcal{A}_{1}
\\
\text{\emph{entity infomorphism}}
& \mathcal{E}_{2}
& \mathrmbfit{ent}(\mathcal{M}_{2})
& \xrightleftharpoons{{\langle{r,k}\rangle}}
& \mathrmbfit{ent}(\mathcal{M}_{1})
& \mathcal{E}_{1}\;
\\
\text{\emph{schema morphism}}
& \mathcal{S}_{2}
& \mathrmbfit{sch}(\mathcal{M}_{2})
& \xRightarrow{{\langle{r,f}\rangle}}
& \mathrmbfit{sch}(\mathcal{M}_{1})
& \mathcal{S}_{1}
\\
\text{\emph{universe morphism}}
& \mathcal{U}_{2}
& \mathrmbfit{univ}(\mathcal{M}_{2})
& \xLeftarrow{{\langle{k,g}\rangle}}
& \mathrmbfit{univ}(\mathcal{M}_{1})
& \mathcal{U}_{1}
\end{array}$}}
\end{center}
satisfying the conditions
\begin{flushleft}
{\footnotesize{$\setlength{\extrarowheight}{4pt}{
\begin{array}{@{\hspace{38pt}}r@{\hspace{20pt}}c}
\text{\emph{typed domain morphism}}
& 
g(y_{1}){\;\models_{\mathcal{A}_{2}}\;}x_{2}
{\;\;\;\text{\underline{iff}}\;\;\;}
y_{1}{\;\models_{\mathcal{A}_{1}}\;}f(x_{2})
\\
\text{\emph{entity infomorphism}}
&
k(k_{1}){\;\models_{\mathcal{E}_{2}}\;}r_{2}
{\;\;\;\text{\underline{iff}}\;\;\;}
k_{1}{\;\models_{\mathcal{E}_{1}}\;}r(r_{2})\;.
\\
\text{\emph{schema morphism}}
&
r{\;\cdot\;}\sigma_{1}
{\;\;=\;\;}
\sigma_{2}{\;\cdot\;}{\scriptstyle\sum}_{f}
\\
\text{\emph{universe morphism}}
&
k{\;\cdot\;}\tau_{2}
{\;\;=\;\;}
\tau_{1}{\;\cdot\;}{\scriptstyle\sum}_{g}
\end{array}}$}}
\end{flushleft}
\begin{figure}
\begin{center}
{{\begin{tabular}{c}
\setlength{\unitlength}{0.6pt}
\begin{picture}(320,250)(-115,-50)
\put(-30,60){\begin{picture}(0,0)(0,0)
\put(0,80){\makebox(0,0){\scriptsize{$R_{2}$}}}
\put(120,80){\makebox(0,0){\scriptsize{$R_{1}$}}}
\put(0,0){\makebox(0,0){\scriptsize{$K_{2}$}}}
\put(120,0){\makebox(0,0){\scriptsize{$K_{1}$}}}
\put(60,88){\makebox(0,0){\scriptsize{$r$}}}
\put(60,8){\makebox(0,0){\scriptsize{$k$}}}
\put(-2,40){\makebox(0,0)[r]{\scriptsize{$\scriptscriptstyle\models_{\mathcal{E}_{2}}$}}}
\put(118,45){\makebox(0,0)[r]{\scriptsize{$\scriptscriptstyle\models_{\mathcal{E}_{1}}$}}}
\put(20,80){\vector(1,0){80}}
\put(100,0){\vector(-1,0){80}}
\put(0,65){\line(0,-1){50}}
\put(120,65){\line(0,-1){50}}
\end{picture}}
\put(0,0){\begin{picture}(0,0)(0,0)
\put(0,80){\makebox(0,0){\scriptsize{$\mathrmbf{List}(X_{2})$}}}
\put(120,80){\makebox(0,0){\scriptsize{$\mathrmbf{List}(X_{1})$}}}
\put(0,0){\makebox(0,0){\scriptsize{$\mathrmbf{List}(Y_{2})$}}}
\put(120,0){\makebox(0,0){\scriptsize{$\mathrmbf{List}(Y_{1})$}}}
\put(60,88){\makebox(0,0){\scriptsize{${\scriptstyle\sum}_{f}$}}}
\put(60,8){\makebox(0,0){\scriptsize{${\scriptstyle\sum}_{g}$}}}
\put(5,36){\makebox(0,0)[l]{\scriptsize{$\models_{\mathrmbf{List}(\mathcal{A}_{2})}$}}}
\put(125,36){\makebox(0,0)[l]{\scriptsize{$\models_{\mathrmbf{List}(\mathcal{A}_{1})}$}}}
\put(30,80){\vector(1,0){57}}
\put(90,0){\vector(-1,0){60}}
\put(0,65){\line(0,-1){50}}
\put(120,65){\line(0,-1){50}}
\end{picture}}
\put(-5,110){\makebox(0,0)[l]{\scriptsize{$\sigma_{2}$}}}
\put(115,110){\makebox(0,0)[l]{\scriptsize{$\sigma_{1}$}}}
\put(-20,30){\makebox(0,0)[r]{\scriptsize{$\tau_{2}$}}}
\put(100,30){\makebox(0,0)[r]{\scriptsize{$\tau_{1}$}}}
\put(-23,130){\vector(1,-2){18}}
\put(97,130){\vector(1,-2){18}}
\put(-23,50){\vector(1,-2){18}}
\put(97,50){\vector(1,-2){18}}
\put(30,190){\makebox(0,0){\scriptsize{$\overset{\textstyle{\stackrel{\text{\emph{schema morphism}}}{
\mathrmbfit{sch}(\mathcal{M}_{2})\xRightarrow{\;{\langle{r,f}\rangle}\;}\mathrmbfit{sch}(\mathcal{M}_{1})}}}{\overbrace{\rule{100pt}{0pt}}}$}}}
\put(-80,100){\makebox(0,0)[r]{\scriptsize{$\underset{\scriptstyle{\text{\emph{entity infomorphism}}}}
{\mathcal{E}_{2}\xrightleftharpoons{{\langle{r,k}\rangle}}\mathcal{E}_{1}}\left\{\rule{0pt}{30pt}\right.$}}}
\put(180,40){\makebox(0,0)[l]{\scriptsize{
$\left.\rule{0pt}{30pt}\right\}\underset{\scriptstyle{\text{\emph{typed domain morphism}}}}
{\mathrmbf{List}\bigl(\mathcal{A}_{2}\xrightleftharpoons{{\langle{f,g}\rangle}}\mathcal{A}_{1}\bigr)}$}}}
\put(60,-45){\makebox(0,0){\scriptsize{$\underset{\textstyle{\stackrel{\textstyle{
\mathrmbfit{univ}(\mathcal{M}_{2})\xLeftarrow{\;{\langle{k,g}\rangle}\;}\mathrmbfit{univ}(\mathcal{M}_{1})
}}{\scriptstyle{\text{\emph{universe morphism}}}}}}{\underbrace{\rule{100pt}{0pt}}}$}}}
\end{picture}
\\ \\
\end{tabular}}}
\end{center}
\caption{Structure Morphism}
\label{fig:fole:struc:mor}
\end{figure}
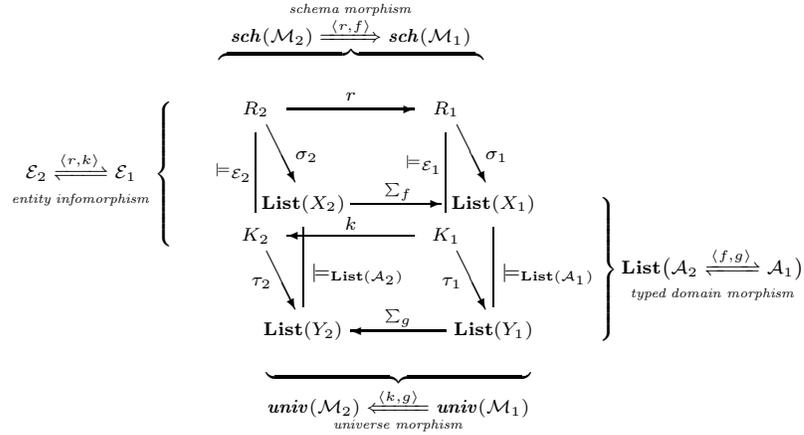
The designation defining condition states that
for any $k_{1}{\,\in\,}K_{1}$ and $r_{2}{\,\in\,}R_{2}$,
\vspace{-10pt}
\begin{center}
{\footnotesize{$\begin{array}{c}
\Bigl(
k(k_{1}){\;\models_{\mathcal{E}_{2}}\;}r_{2}
\;\text{\underline{iff}}\;
k_{1}{\;\models_{\mathcal{E}_{1}}\;}r(r_{2})
\Bigr)
\;\text{\underline{implies}}\;
\\
\Bigl(
\tau_{2}(k(k_{1}))={\scriptstyle\sum}_{g}(\tau_{1}(k_{1})){\;\models_{\mathrmbf{List}(\mathcal{A}_{2})}\;}\sigma_{2}(r_{2})
\;\text{\underline{iff}}\;
\tau_{1}(k_{1}){\;\models_{\mathrmbf{List}(\mathcal{A}_{1})}\;}{\scriptstyle\sum}_{f}(\sigma_{2}(r_{2}))=\sigma_{1}(r(r_{2}))
\Bigr).
\end{array}$}}
\end{center}
Structure morphisms compose component-wise.
Let $\mathrmbf{Struc}$ denote the context of structures and structure morphisms.
In the appendix \S~\ref{sub:sec:struc:fbr:pass},
we develop $\mathrmbf{Struc}$ as a \emph{fibered} mathematical context 
in two orientations: 
either as the Grothedieck construction of 
the schema indexed mathematical context of structures 
$\mathrmbf{Sch}^{\mathrm{op}}{\!\xrightarrow{\mathrmbfit{struc}}\;}\mathrmbf{Cxt}$
or as the Grothedieck construction of 
the universe indexed mathematical context of structures 
$\mathrmbf{Univ}^{\mathrm{op}}{\!\xrightarrow{\mathrmbfit{struc}}\;}\mathrmbf{Cxt}$.
The schema indexed mathematical context of structures is used 
to establish the institutional aspect of {\ttfamily FOLE}.

Any structure morphism
$\mathcal{M}_{2}\xrightleftharpoons{{\langle{r,k,f,1_{Y}}\rangle}}\mathcal{M}_{1}$
with identity value map $Y\xrightarrow{1_{Y}}Y$
has the interpretive presentation in Fig.~\ref{fig:interp:struc:mor}.
\footnote{Any infomorphism $\mathcal{A}_{2}\xrightleftharpoons{{\langle{f,g}\rangle}}\mathcal{A}_{1}$
has the equivalent condition
$f{\,\cdot\,}\mathrmbfit{ext}_{\mathcal{A}_{1}}=\mathrmbfit{ext}_{\mathcal{A}_{2}}{\,\cdot\,}g^{-1}$.}
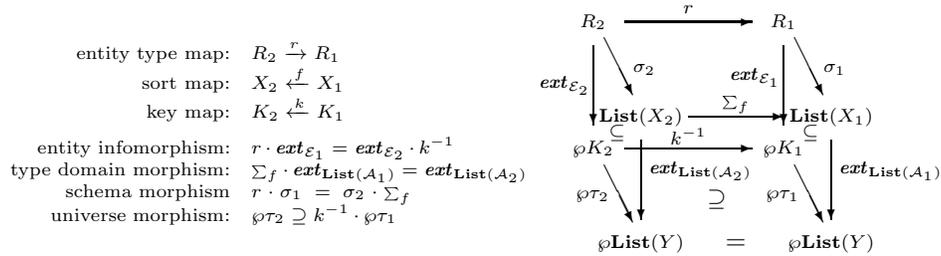
\begin{figure}
\begin{center}
{{\begin{tabular}{@{\hspace{-55pt}}c@{\hspace{25pt}}c}
{{\scriptsize{$\begin{array}[b]{r@{\hspace{8pt}}l}
\text{entity type map:}
& 
R_{2}\xrightarrow{r}R_{1}
\\
\text{sort map:}
& 
X_{2}\xleftarrow{f}X_{1}
\\
\text{key map:}
& 
K_{2}\xleftarrow{k}K_{1}
\\
\rule[5pt]{0pt}{6pt}
\text{entity infomorphism:}
& 
r{\,\cdot\,}\mathrmbfit{ext}_{\mathcal{E}_{1}}=\mathrmbfit{ext}_{\mathcal{E}_{2}}{\,\cdot\,}k^{-1}
\\
\text{type domain morphism:}
& 
{\scriptstyle\sum}_{f}{\,\cdot\,}\mathrmbfit{ext}_{\mathrmbf{List}(\mathcal{A}_{1})}=\mathrmbfit{ext}_{\mathrmbf{List}(\mathcal{A}_{2})}
\\
\text{schema morphism}
&
r{\;\cdot\;}\sigma_{1}{\;\;=\;\;}\sigma_{2}{\;\cdot\;}{\scriptstyle\sum}_{f}
\\
\text{universe morphism:}
& 
{\wp}\tau_{2}{\;\supseteq\;}k^{-1}{\,\cdot\,}{\wp}\tau_{1}
\end{array}$}}}
&
{{\begin{tabular}{c}
\setlength{\unitlength}{0.6pt}
\begin{picture}(120,20)(-20,20)
\put(-30,60){\begin{picture}(0,0)(0,0)
\put(0,80){\makebox(0,0){\scriptsize{$R_{2}$}}}
\put(120,80){\makebox(0,0){\scriptsize{$R_{1}$}}}
\put(0,0){\makebox(0,0){\scriptsize{${\wp}K_{2}$}}}
\put(120,0){\makebox(0,0){\scriptsize{${\wp}K_{1}$}}}
\put(60,88){\makebox(0,0){\scriptsize{$r$}}}
\put(60,8){\makebox(0,0){\scriptsize{$k^{-1}$}}}
\put(-2,40){\makebox(0,0)[r]{\scriptsize{$\mathrmbfit{ext}_{\mathcal{E}_{2}}$}}}
\put(118,45){\makebox(0,0)[r]{\scriptsize{$\mathrmbfit{ext}_{\mathcal{E}_{1}}$}}}
\put(20,80){\vector(1,0){80}}
\put(20,0){\vector(1,0){80}}
\put(0,65){\vector(0,-1){50}}
\put(120,65){\vector(0,-1){50}}
\end{picture}}
\put(0,0){\begin{picture}(0,0)(0,0)
\put(0,80){\makebox(0,0){\scriptsize{$\mathrmbf{List}(X_{2})$}}}
\put(120,80){\makebox(0,0){\scriptsize{$\mathrmbf{List}(X_{1})$}}}
\put(0,0){\makebox(0,0){\scriptsize{${\wp}\mathrmbf{List}(Y)$}}}
\put(120,0){\makebox(0,0){\scriptsize{${\wp}\mathrmbf{List}(Y)$}}}
\put(60,88){\makebox(0,0){\scriptsize{${\scriptstyle\sum}_{f}$}}}
\put(60,0){\makebox(0,0){\normalsize{$=$}}}
\put(5,46){\makebox(0,0)[l]{\scriptsize{$\mathrmbfit{ext}_{\mathrmbf{List}(\mathcal{A}_{2})}$}}}
\put(125,46){\makebox(0,0)[l]{\scriptsize{$\mathrmbfit{ext}_{\mathrmbf{List}(\mathcal{A}_{1})}$}}}
\put(30,80){\vector(1,0){60}}
\put(0,65){\vector(0,-1){50}}
\put(120,65){\vector(0,-1){50}}
\end{picture}}
\put(-5,110){\makebox(0,0)[l]{\scriptsize{$\sigma_{2}$}}}
\put(115,110){\makebox(0,0)[l]{\scriptsize{$\sigma_{1}$}}}
\put(-20,30){\makebox(0,0)[r]{\scriptsize{${\wp}\tau_{2}$}}}
\put(100,30){\makebox(0,0)[r]{\scriptsize{${\wp}\tau_{1}$}}}
\put(-15,70){\makebox(0,0){\scriptsize{$\subseteq$}}}
\put(107,70){\makebox(0,0){\scriptsize{$\subseteq$}}}
\put(48,28){\makebox(0,0){\footnotesize{$\supseteq$}}}
\put(-23,130){\vector(1,-2){18}}
\put(97,130){\vector(1,-2){18}}
\put(-23,50){\vector(1,-2){18}}
\put(97,50){\vector(1,-2){18}}
\end{picture}
\\ \\
\end{tabular}}}
\end{tabular}}}
\end{center}
\caption{Interpreted Structure Morphism}
\label{fig:interp:struc:mor}
\end{figure}

\newpage

\begin{definition}(key embedding)\label{def:key:embed:mor}
Any structure morphism
$\mathcal{M}_{2}\xrightleftharpoons{{\langle{r,k,f,g}\rangle}}\mathcal{M}_{1}$
has a companion key embedding structure morphism
%
\[\mbox
{\footnotesize{$
\dot{\mathcal{M}}_{2} = {\langle{\mathcal{E}_{2},\dot{\sigma}_{2},\dot{\tau}_{2},\mathcal{E}_{2}{+}\mathcal{A}_{2}}\rangle}
\xrightleftharpoons{{\langle{r,k,f,f{+}g}\rangle}}
{\langle{\mathcal{E}_{1},\dot{\sigma}_{1},\dot{\tau}_{1},\mathcal{E}_{1}{+}\mathcal{A}_{1}}\rangle} = \dot{\mathcal{M}}_{1}
$.}\normalsize}
\]
with the following components:
\vspace{-10pt}
\begin{flushleft}
{\footnotesize{$\setlength{\extrarowheight}{4pt}{
\begin{array}{@{\hspace{20pt}}r@{\hspace{20pt}}c}
\text{\emph{typed domain morphism}}
& 
\mathcal{E}_{2}{+}\mathcal{A}_{2}\xrightleftharpoons{{\langle{r{+}f,k{+}g}\rangle}}\mathcal{E}_{1}{+}\mathcal{A}_{1}
\\
\text{\emph{entity infomorphism}}
&
\mathcal{E}_{2}\xrightleftharpoons{{\langle{r,k}\rangle}}\mathcal{E}_{1}
\\
\text{\emph{schema morphism}}
&
\dot{\mathcal{S}}_{2}={\langle{R_{2},\dot{\sigma}_{2},R_{2}{+}X_{2}}\rangle}
\xRightarrow{{\langle{r,r{+}f}\rangle}}
{\langle{R_{1},\dot{\sigma}_{1},R_{1}{+}X_{1}}\rangle}=\mathcal{S}_{1}
\\
\text{\emph{universe morphism}}
&
\mathcal{U}_{2}={\langle{K_{2},\dot{\tau}_{2},K_{2}{+}\,Y_{2}}\rangle}
\xLeftarrow{{\langle{k,k{+}g}\rangle}}
{\langle{K_{1},\dot{\tau}_{1},K_{1}{+}\,Y_{1}}\rangle}=\mathcal{U}_{1}
\end{array}}$}}
\end{flushleft}
\end{definition}
\begin{proof}
The following conditions must hold.
\begin{flushleft}
{\footnotesize{$\setlength{\extrarowheight}{4pt}{
\begin{array}{@{\hspace{25pt}}r@{\hspace{20pt}}c}
\text{\emph{typed domain morphism}}
& 
\overset{\textstyle{k(k_{1})}}{g(y_{1})}
{\;\models_{\mathcal{E}_{2}{+}\mathcal{A}_{2}}\;}
\overset{\textstyle{r_{2}}\rule[-3pt]{0pt}{3pt}}{x_{2}}
{\;\;\;\text{\underline{iff}}\;\;\;}
\overset{\textstyle{k_{1}}\rule[-3pt]{0pt}{3pt}}{y_{1}}
{\;\models_{\mathcal{E}_{1}{+}\mathcal{A}_{1}}\;}
\overset{\textstyle{r(r_{1})}}{f(x_{2})}\bigr)
\\
\text{\emph{entity infomorphism}}
&
k(k_{1}){\;\models_{\mathcal{E}_{2}}\;}r_{2}
{\;\;\;\text{\underline{iff}}\;\;\;}
k_{1}{\;\models_{\mathcal{E}_{1}}\;}r(r_{2})\;.
\\
\text{\emph{schema morphism}}
&
r{\;\cdot\;}\dot{\sigma}_{1}
{\;\;=\;\;}
\dot{\sigma}_{2}{\;\cdot\;}{\scriptstyle\sum}_{r{+}f}
\\
\text{\emph{universe morphism}}
&
k{\;\cdot\;}\dot{\tau}_{2}
{\;\;=\;\;}
\dot{\tau}_{1}{\;\cdot\;}{\scriptstyle\sum}_{k{+}g}
\end{array}}$}}
\end{flushleft}
%
We use the comparable conditions for the original structure morphism
$\mathcal{M}_{2}\xrightleftharpoons{{\langle{r,k,f,g}\rangle}}\mathcal{M}_{1}$.
The entity infomorphism condition is given.
The type domain morphism condition is straightforward.
We show the schema morphism condition.
The universe morphism condition is similar.
The schema morphism condition for the original structure morphism is
$r{\;\cdot\;}\sigma_{1} = \sigma_{2}{\;\cdot\;}{\scriptstyle\sum}_{f}$;
that is,
for any $r_{2}{\,\in\,}R_{2}$,
if $\sigma_{2}(r_{2}){\;=\;}{\langle{I_{2},s_{2}}\rangle}$
and $\sigma_{1}(r(r_{2})){\;=\;}{\langle{I_{1},s_{1}}\rangle}$,
then ${\langle{I_{1},s_{1}}\rangle}{\;=\;}{\scriptstyle\sum}_{f}(I_{2},s_{2})$.
Hence,
the schema morphism condition for the key-embedding structure morphism holds,
since
$\dot{\sigma}_{1}(r(r_{2}))
{\;\;=\;\;}
{\langle{1,r(r_{2})}\rangle}{+}{\langle{I_{1},s_{1}}\rangle}
{\;\;=\;\;}
{\scriptstyle\sum}_{r}(1,r_{2}){+}
{\scriptstyle\sum}_{f}(I_{2},s_{2})
\newline
{\;\;=\;\;}
{\scriptstyle\sum}_{r{+}f}\bigl({\langle{1,r_{2}}\rangle}{+}{\langle{I_{2},s_{2}}\rangle}\bigr)
{\;\;=\;\;}
{\scriptstyle\sum}_{r{+}f}(\dot{\sigma}_{2}(r_{2}))
$
\rule{5pt}{5pt}
\end{proof}

%
\paragraph{Integrity Constraints.}
%

Integrity constraints help preserve the validity and consistency of data.
Here we briefly explain how various integrity constraints 
are represented in the {\ttfamily ERA} data model of {\ttfamily FOLE}.
\begin{description}
\item[Entity:] (primary key rule)
Entity integrity states that 
every table must have a primary key and 
that the column or columns chosen to be the primary key should be unique and not null.
In the {\ttfamily ERA} data model of {\ttfamily FOLE},
entity integrity 
asserts that
the universe $\mathcal{U} = {\langle{K,\tau,Y}\rangle}$
of a structure $\mathcal{M}$ is well-defined.
%
\item[Domain:] 
Domain integrity specifies that all columns in a relational database must be declared upon a defined domain. 
In the {\ttfamily ERA} data model of {\ttfamily FOLE},
domain integrity 
asserts that
the schema $\mathcal{S} = {\langle{R,\sigma,X}\rangle}$ and
the list designation ${\langle{\sigma,\tau}\rangle} : \mathcal{E} \rightrightarrows \mathrmbf{List}(\mathcal{A})$
of a structure $\mathcal{M}$ 
are well-defined.
%
\item[Referential:] (foreign key rule) 
Referential integrity states that 
the foreign-key value of a source table refers to a primary key value of a target table.
In the {\ttfamily ERA} data model of {\ttfamily FOLE},
referential integrity asserts that
the {\ttfamily ERA} data model of {\ttfamily FOLE}
is a mixed data model.
%
\end{description}
%

\paragraph{Algebra.}
%

For simplicity of presentation, 
this paper and the paper on {\ttfamily FOLE} superstructure (Kent~\cite{kent:fole:era:supstruc})
use a simplified form of {\ttfamily FOLE}, 
in contrast to the full form presented in Kent~\cite{kent:iccs2013}.
In this paper and Kent \cite{kent:fole:era:supstruc}, 
schemas are used in place of (many-sorted) first-order logical languages. 
Schemas are simplified logical languages without function symbols.
The main practical result is that 
signature morphisms ${\langle{I',s'}\rangle}\xrightarrow{h}{\langle{I,s}\rangle}$
are replaced by
term vectors  ${\langle{I',s'}\rangle} \xrightharpoondown{t} {\langle{I,s}\rangle}$
in the full version of {\ttfamily FOLE}.
Signature morphisms are simplified term vectors without function symbols.
In the full version of {\ttfamily FOLE}, 
equations can be defined between parallel pairs of term vectors 
${\langle{I',s'}\rangle} \xrightharpoondown{\,t,t'\!} {\langle{I,s}\rangle}$,
thus allowing the use of equational presentations and their congruences.
\footnote{The tuple relational calculus is a query language for relational databases.
In order to use the tuple calculus in the {\ttfamily FOLE},
we need to enrich with many-sorted constant declarations and equational presentations.
Constant declarations are first-order logical languages 
with sorted nullary function symbols.
A constant $c$ of sort $x$ is an $x$-sorted nullary function symbol
$x \xrightharpoondown{c} {\langle{\emptyset,0_{X}}\rangle}$.}
Also in the full version of {\ttfamily FOLE}, 
the tuple map along signature morphisms becomes the algebraic operation along term vectors;
hence,
formula flow (substitution/quantification) 
in Kent \cite{kent:fole:era:supstruc}
is lifted from being along signature morphisms to being along term vectors.
\footnote{Let {\ttfamily FOLE-ARCH} denote Fig.~1.~in Kent~\cite{kent:iccs2013}.
{\ttfamily FOLE-ARCH} is the 3-dimensional visualization of the fibered architecture of {\ttfamily FOLE}.
The upper right quadrant of Fig.~\ref{fig:cxt:struc}
corresponds to the the 2-D prism below $\mathrmbf{Rel}$ in {\ttfamily FOLE-ARCH}.
As indicated in {\ttfamily FOLE-ARCH},
to move 
from the simple version of the {\ttfamily FOLE} foundation used here (Fig.~\ref{fig:cxt:struc})
to the full version in {\ttfamily FOLE-ARCH},
we lift from sort sets to algebraic languages and from typed domains to many-sorted algebras.}
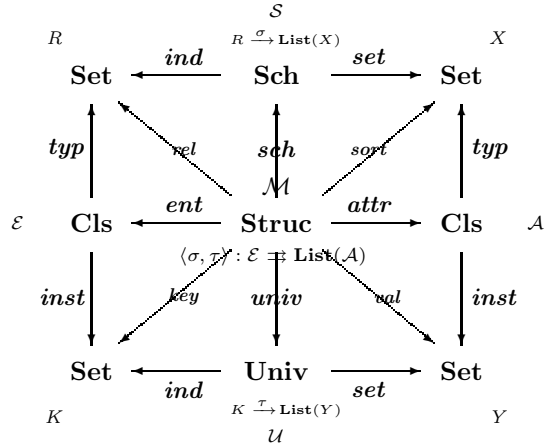
\begin{figure}
\begin{center}
{{\begin{tabular}{c}
\setlength{\unitlength}{0.7pt}
\begin{picture}(200,220)(0,-20)
\put(0,160){\makebox(0,0){\normalsize{$\mathrmbf{Set}$}}}
\put(100,160){\makebox(0,0){\normalsize{$\mathrmbf{Sch}$}}}
\put(200,160){\makebox(0,0){\normalsize{$\mathrmbf{Set}$}}}
\put(0,80){\makebox(0,0){\normalsize{$\mathrmbf{Cls}$}}}
\put(100,80){\makebox(0,0){\normalsize{$\mathrmbf{Struc}$}}}
\put(200,80){\makebox(0,0){\normalsize{$\mathrmbf{Cls}$}}}
\put(0,0){\makebox(0,0){\normalsize{$\mathrmbf{Set}$}}}
\put(100,0){\makebox(0,0){\normalsize{$\mathrmbf{Univ}$}}}
\put(200,0){\makebox(0,0){\normalsize{$\mathrmbf{Set}$}}}
\put(50,170){\makebox(0,0){\footnotesize{$\mathrmbfit{ind}$}}}
\put(150,170){\makebox(0,0){\footnotesize{$\mathrmbfit{set}$}}}
\put(50,-10){\makebox(0,0){\footnotesize{$\mathrmbfit{ind}$}}}
\put(150,-10){\makebox(0,0){\footnotesize{$\mathrmbfit{set}$}}}
\put(50,90){\makebox(0,0){\footnotesize{$\mathrmbfit{ent}$}}}
\put(150,90){\makebox(0,0){\footnotesize{$\mathrmbfit{attr}$}}}
\put(-3,120){\makebox(0,0)[r]{\footnotesize{$\mathrmbfit{typ}$}}}
\put(100,120){\makebox(0,0){\footnotesize{$\mathrmbfit{sch}$}}}
\put(205,120){\makebox(0,0)[l]{\footnotesize{$\mathrmbfit{typ}$}}}
\put(-3,40){\makebox(0,0)[r]{\footnotesize{$\mathrmbfit{inst}$}}}
\put(100,40){\makebox(0,0){\footnotesize{$\mathrmbfit{univ}$}}}
\put(205,40){\makebox(0,0)[l]{\footnotesize{$\mathrmbfit{inst}$}}}
\put(50,120){\makebox(0,0){\scriptsize{$\mathrmbfit{rel}$}}}
\put(150,120){\makebox(0,0){\scriptsize{$\mathrmbfit{sort}$}}}
\put(50,40){\makebox(0,0){\scriptsize{$\mathrmbfit{key}$}}}
\put(160,40){\makebox(0,0){\scriptsize{$\mathrmbfit{val}$}}}
\put(70,160){\vector(-1,0){48}}
\put(130,160){\vector(1,0){48}}
\put(70,80){\vector(-1,0){48}}
\put(130,80){\vector(1,0){48}}
\put(70,0){\vector(-1,0){48}}
\put(130,0){\vector(1,0){48}}
\put(0,94){\vector(0,1){50}}
\put(100,94){\vector(0,1){50}}
\put(200,94){\vector(0,1){50}}
\put(0,64){\vector(0,-1){50}}
\put(100,64){\vector(0,-1){50}}
\put(200,64){\vector(0,-1){50}}
\qbezier[50](75,95)(45,120)(15,145)\put(15,145){\vector(-4,3){0}}
\qbezier[50](125,95)(155,120)(185,145)\put(185,145){\vector(4,3){0}}
\qbezier[50](75,65)(45,40)(15,15)\put(15,15){\vector(-4,-3){0}}
\qbezier[50](125,65)(155,40)(185,15)\put(185,15){\vector(4,-3){0}}
\put(-20,180){\makebox(0,0){\scriptsize{$\mathrmit{R}$}}}
\put(100,195){\makebox(0,0){\scriptsize{$\mathcal{S}$}}}
\put(105,180){\makebox(0,0){\tiny{$R\xrightarrow{\sigma}\!\mathrmbf{List}(X)$}}}
\put(220,180){\makebox(0,0){\scriptsize{$\mathrmit{X}$}}}
\put(-40,80){\makebox(0,0){\scriptsize{$\mathcal{E}$}}}
\put(100,100){\makebox(0,0){\footnotesize{$\mathcal{M}$}}}
\put(98,62){\makebox(0,0){\scriptsize{${\langle{\sigma,\tau}\rangle}:\mathcal{E}\rightrightarrows\mathrmbf{List}(\mathcal{A})$}}}
\put(240,80){\makebox(0,0){\scriptsize{$\mathcal{A}$}}}
\put(-20,-25){\makebox(0,0){\scriptsize{$\mathrmit{K}$}}}
\put(105,-20){\makebox(0,0){\tiny{$K\xrightarrow{\tau}\!\mathrmbf{List}(Y)$}}}
\put(100,-35){\makebox(0,0){\scriptsize{$\mathcal{U}$}}}
\put(220,-25){\makebox(0,0){\scriptsize{$\mathrmit{Y}$}}}
\end{picture}
\\\\
\end{tabular}}}
\end{center}
\caption{{\ttfamily FOLE} Foundation}
\label{fig:cxt:struc}
\end{figure}
%

\newpage
\subsection{Interpretation.}\label{sub:sub:sec:data:model:interp}

In the model theory for traditional many-sorted first-order logic,
a (possible world, model) structure corresponds to an interpretation 
of relation symbols (entity types) in terms of relations in a typed domain.
The {\ttfamily FOLE} approach to logic 
replaces $n$-tuples with lists, 
defines quantification/substitution along signature morphisms 
(Kent~\cite{kent:fole:era:supstruc})
(or term vectors in the full version~\cite{kent:iccs2013}),
and following databases, 
incorporates identifiers (keys) for data value lists (tuples)
(here and in 
the paper on {\ttfamily FOLE} tables 
Kent~\cite{kent:fole:era:tbl}).
The {\ttfamily FOLE} approach modifies the idea of model-theoretic interpretation as follows.
%
\footnote{Hence,
the notions of
(1) a many-sorted first-order logic interpretation,
(2) a {\ttfamily FOLE} structure, and 
(3) an {\ttfamily ERA} data model are all equivalent;
and each can be implemented as a relational database with associated logic
.}
%

We assume that the traditional many-sorted first-order logic language is represented by the schema
$\mathcal{S} = {\langle{R,\sigma,X}\rangle}$
and that the typed domain is represented by the attribute classification 
$\mathcal{A} = {\langle{X,Y,\models_{\mathcal{A}}}\rangle}$.
%
We further assume that these are components of a structure 
$\mathcal{M} = {\langle{\mathcal{E},{\langle{\sigma,\tau}\rangle},\mathcal{A}}\rangle}$.
\begin{description}
\item[Traditional:] 
In the traditional approach,
an entity type $r \in R$ 
is interpreted as the set of descriptors of entities in the extent of $r$.
For $X$-signature $\sigma(r) = {\langle{I,s}\rangle} \in \mathrmbf{List}(X)$,
this is the subset of $Y$-tuples
$\mathrmbfit{I}_{\mathcal{M}}(r) = {\wp}\tau(\mathrmbfit{ext}_{\mathcal{E}}(r))
\in 
{\wp}\mathrmbfit{tup}_{\mathcal{A}}(I,s) 
= {\wp}\mathrmbfit{ext}_{\mathrmbf{List}(\mathcal{A})}(I,s)$,
an element of the fiber relational order
${\langle{\mathrmbf{Rel}_{\mathcal{A}}(I,s),\subseteq}\rangle}$.
\footnote{The fibered context $\mathrmbf{Rel}(\mathcal{A})$ is defined in 
the paper on {\ttfamily FOLE} tables 
(Kent~\cite{kent:fole:era:tbl}).
An object of $\mathrmbf{Rel}(\mathcal{A})$,
called an $\mathcal{A}$-relation,
is a pair ${\langle{I,s,R}\rangle}$
consisting of an indexing $X$-signature ${\langle{I,s}\rangle}$
and a subset of $\mathcal{A}$-tuples 
$R \in {\wp}\mathrmbfit{tup}_{\mathcal{A}}(I,s) 
= {\wp}\mathrmbfit{ext}_{\mathrmbf{List}(\mathcal{A})}(I,s)$.}
%
This defines the traditional interpretation function
\begin{equation}\label{eqn:rel:ent:typ}
R\xrightarrow{\;\mathrmbfit{I}_{\mathcal{M}}\;}\mathrmbf{Rel}(\mathcal{A})\,.
\end{equation}
For all ${r}\in{R}$,
we have the relationships
\begin{equation}\label{ext:interp:equiv:ent:typ}
{\setlength{\extrarowheight}{2pt}\begin{array}{c@{\hspace{20pt}}c}
{\wp}\tau(\mathrmbfit{ext}_{\mathcal{E}}(r)){\;=\;}\mathrmbfit{I}_{\mathcal{M}}(r)
& 
\mathrmbfit{ext}_{\mathcal{E}}(r){\;\subseteq\;}\tau^{-1}(\mathrmbfit{I}_{\mathcal{M}}(r))\,.
\end{array}}
\end{equation}
%
%
\end{description}
\begin{definition}\label{assump:extens:struc}
The inequality
$\mathrmbfit{ext}_{\mathcal{E}}(r){\;\subseteq\;}\tau^{-1}(\mathrmbfit{I}_{\mathcal{M}}(r))$
says that
$\mathrmbfit{ext}_{\mathcal{E}}(r)$ is not the morphic closure of itself
w.r.t.\ the tuple map
$K\xrightarrow{\tau}\mathrmbf{List}(Y)$.
A structure $\mathcal{M}$ is called \emph{extensive}
when the right hand expression in (\ref{ext:interp:equiv:ent:typ})
is an equality:
$\mathrmbfit{ext}_{\mathcal{E}}(r){\;=\;}\tau^{-1}(\mathrmbfit{I}_{\mathcal{M}}(r))$
for any entity type ${r} \in {R}$.
\footnote{\underline{Philosophical note}:\label{ftnt:phil:note}
In the knowledge resources for a community, 
the entities are of first importance.
The tuples in $\mathrmbf{List}(Y)$
are descriptors,
which may or may not have an identity.
An entity consists of an identifier $k\in{K}$ and its descriptor $\tau(k)\in{\wp}\tau(K)$.
Tuples with identity are those in
${\wp}\tau(K){\;\subseteq\;}\mathrmbf{List}(Y)$.
Two entities that have the same descriptor are said to be ``descriptor-equivalent''.}
\end{definition}
Any structure $\mathcal{M}$ with an injective tuple map $K\xrightarrow{\tau}\mathrmbf{List}(Y)$
has an associated extensive structure.
An example is the key-embedding structure $\dot{\mathcal{M}}$.
%
%
\begin{description}
\comment{
\item[Relational:] 
The relational interpretation function
\begin{equation}\label{eqn:inc:rel:ent:typ}
R\xrightarrow[\mathrmbfit{T}_{\mathring{\mathcal{M}}}]{\;\mathrmbfit{R}_{\mathcal{M}}\;}|\mathrmbf{Tbl}(\mathcal{A})|
\end{equation}
maps an entity type $r \in {R}$
with signature $\sigma(r) = {\langle{I,s}\rangle}$
to the ${\langle{I,s}\rangle}$-indexed $\mathcal{A}$-table
$\mathrmbfit{R}_{\mathcal{M}}(r) = {\langle{\mathrmbfit{I}_{\mathcal{M}}(r),inc}\rangle}$
consisting of 
the traditional interpretation (key set) $\mathrmbfit{I}_{\mathcal{M}}(r)$ and 
the inclusion tuple function $\mathrmbfit{I}_{\mathcal{M}}(r)\xhookrightarrow{inc}\mathrmbfit{tup}_{\mathcal{A}}(I,s)$.
%
\footnote{Two tables are informationally equivalent
when they contain the same information; that is,
when their image relations are equivalent in 
$\mathrmbf{Rel}_{\mathcal{A}}(I,s)={\wp}\mathrmbfit{tup}_{\mathcal{A}}(I,s)$.
In particular,
the tables $\mathrmbfit{T}_{\mathcal{M}}(r)$ and $\mathrmbfit{R}_{\mathcal{M}}(r)$
are informationally equivalent
for any entity type $r \in {R}$.}
\newline
}

\newpage

\item[Tabular:]
In the database approach,
an entity type $r \in R$ 
is interpreted as a table
with the entities (both the keys and their descriptors) being explicit.
For $X$-signature $\sigma(r) = {\langle{I,s}\rangle} \in \mathrmbf{List}(X)$,
this is the
${\langle{I,s}\rangle}$-indexed $\mathcal{A}$-table
\begin{center}
\begin{tabular}{c@{\hspace{25pt}}c}
{\setlength{\extrarowheight}{2pt}{\scriptsize{$\begin{array}{c@{\hspace{20pt}}c}
K(r)\xrightarrow{\tau_{r}}\mathrmbfit{tup}_{\mathcal{A}}(I,s)
\\
\tau_{r}(k) = {\langle{I,t}\rangle},\;I\xrightarrow{t}Y
\end{array}$}}}
&
{{\begin{tabular}{c}
\setlength{\unitlength}{0.5pt}
\begin{picture}(180,135)(-60,0)
\put(-10,90){\makebox(0,0){\scriptsize{$r$}}}
\put(60,120){\makebox(0,0){\shortstack{\scriptsize{${\langle{I,s}\rangle}$}\\$\overbrace{\rule{50pt}{0pt}}$}}}
\put(-55,40){\makebox(0,0){\scriptsize{$K(r)\left\{\rule{0pt}{20pt}\right.$}}}
\put(25,90){\makebox(0,0){\scriptsize{$\cdots$}}}
\put(60,90){\makebox(0,0){\scriptsize{$i\!:\!s_{i}$}}}
\put(97,90){\makebox(0,0){\scriptsize{$\cdots$}}}
\put(-10,40){\makebox(0,0){\scriptsize{$k$}}}
\put(25,40){\makebox(0,0){\scriptsize{$\cdots$}}}
\put(60,40){\makebox(0,0){\scriptsize{$t_{i}$}}}
\put(97,40){\makebox(0,0){\scriptsize{$\cdots$}}}
\put(-20,78){\line(1,0){140}}
\put(2,0){\line(0,1){100}}
\end{picture}
\end{tabular}}}
\end{tabular}
\end{center}
$\mathrmbfit{T}_{\mathcal{M}}(r) = {\langle{K(r),\tau_{r}}\rangle}$
(visualized above)
consisting of 
the key set $K(r) = \mathrmbfit{ext}_{\mathcal{E}}(r) 
\in {\wp}K$
and 
the (descriptor) tuple function $K(r)\xrightarrow{\tau_{r}}\mathrmbfit{tup}_{\mathcal{A}}(I,s)$,
a restriction 
of the tuple function $K\xrightarrow{\tau}\mathrmbf{List}(Y)$
for the universe $\mathcal{U} = {\langle{K,\tau,Y}\rangle}$.
The tuple function
factors through the traditional interpretation
$\tau_{r} : K(r)\rightarrow\mathrmbfit{I}_{\mathcal{M}}(r)\hookrightarrow\mathrmbfit{tup}_{\mathcal{A}}(I,s)$.
This defines the tabular interpretation function
\begin{equation}\label{eqn:tbl:ent:typ}
R\xrightarrow{\;\mathrmbfit{T}_{\mathcal{M}}\;}\mathrmbf{Tbl}(\mathcal{A})\,.
~\footnote{The fibered context $\mathrmbf{Tbl}(\mathcal{A})$ is defined in 
the paper on {\ttfamily FOLE} tables 
(Kent~\cite{kent:fole:era:tbl}).
An object of $\mathrmbf{Tbl}(\mathcal{A})$,
which is called an $\mathcal{A}$-table,
is a triple ${\langle{I,s,K,t}\rangle}$
consisting of an $X$-signature ${\langle{I,s}\rangle}$,
a set of entities (keys) $K$,
and a tuple function $K\xrightarrow{t}\mathrmbfit{tup}_{\mathcal{A}}(I,s)$
mapping each entity in $K$ to its descriptor $\mathcal{A}$-tuple.}
\end{equation}
The companion key-embedding structure $\dot{\mathcal{M}}$ (Def.~\ref{def:key:embed}) 
is most easily understood from its interpretation tables (visualized below).
\begin{center}
\begin{tabular}{@{\hspace{-30pt}}c@{\hspace{8pt}}c}
{\setlength{\extrarowheight}{2pt}{\scriptsize{$\begin{array}{c@{\hspace{20pt}}c}
K(r)\xrightarrow{\dot{\tau}_{r}}\mathrmbfit{tup}_{\mathcal{E}{+}\mathcal{A}}(1{+}I,r{+}s)
\\
\dot{\tau}_{r}(k) = {\langle{1{+}\,I,k{+}t}\rangle},\;1\xrightarrow{k}K,\;I\xrightarrow{t}Y
\end{array}$}}}
&
{{\begin{tabular}{c}
\setlength{\unitlength}{0.5pt}
\begin{picture}(180,135)(-60,0)
\put(76,120){\makebox(0,0){\shortstack{\scriptsize{${\langle{1{+}I,r{+}s}\rangle}$}\\$\overbrace{\rule{64pt}{0pt}}$}}}
\put(-55,40){\makebox(0,0){\scriptsize{$K(r)\left\{\rule{0pt}{20pt}\right.$}}}
\put(-10,90){\makebox(0,0){\scriptsize{$r$}}}
\put(-10,40){\makebox(0,0){\scriptsize{$k$}}}
\put(25,90){\makebox(0,0){\scriptsize{${\cdot}\!:\!r$}}}
\put(60,90){\makebox(0,0){\scriptsize{$\cdots$}}}
\put(95,90){\makebox(0,0){\scriptsize{$i\!:\!s_{i}$}}}
\put(132,90){\makebox(0,0){\scriptsize{$\cdots$}}}
\put(25,40){\makebox(0,0){\scriptsize{$k$}}}
\put(60,40){\makebox(0,0){\scriptsize{$\cdots$}}}
\put(95,40){\makebox(0,0){\scriptsize{$t_{i}$}}}
\put(132,40){\makebox(0,0){\scriptsize{$\cdots$}}}
\put(-20,78){\line(1,0){140}}
\put(2,0){\line(0,1){100}}
\end{picture}
\end{tabular}}}
\end{tabular}
\end{center}
\end{description}
%

\newpage
\section{Connections}\label{sub:sec:connections}

\subsection{Analogy.}\label{sub:sub:sec:analogy}

In the original paper (Kent~\cite{kent:iccs2013}) 
explaining the first-order logical environment {\ttfamily FOLE},
there was an analogy between 
the top-level ontological categories discussed in (Sowa~\cite{sowa:kr}) and 
the components 
of the first-order logical environment {\ttfamily FOLE}.
We recast that here in terms of the {\ttfamily FOLE} approach to the {\ttfamily ERA} data model.
%
The analogy between 
the Sowa hierarchy and {\ttfamily FOLE}
is illustrated graphically in Fig.~\ref{analogy}.
The twelve categories displayed in the hierarchy 
(we do not concern ourselves here with 
the temporal dimension given by the {\itshape continuent}-{\itshape occurent} distinction)
and the primitives from which they are generated are arranged in the matrix of Tbl.~\ref{matrix}
with the corresponding {\ttfamily FOLE} terminology in parentheses.

The {\itshape physical}-{\itshape abstract} distinction, 
which corresponds to the Heraclitus distinction {\itshape physis}-{\itshape logos},
is represented by
the {\ttfamily FOLE} classification dimension, 
a connection between {\ttfamily FOLE} schemas and {\ttfamily FOLE} universes.
The {\itshape abstract} category corresponds to {\ttfamily FOLE} types,
either entities or attributes.
The {\itshape physical} category corresponds to {\ttfamily FOLE} instances,
either entities or attributes.
The {\itshape independent}-{\itshape relative}-{\itshape mediating} triad 
is represented by
the {\ttfamily FOLE} hypergraph dimension, 
a connection between {\ttfamily FOLE} entities and {\ttfamily FOLE} attributes.
The {\itshape independent} category (firstness) corresponds to {\ttfamily FOLE} entities that are not attributes.
The {\itshape relative} category (secondness) corresponds to {\ttfamily FOLE} attributes that are not entities.
The {\itshape mediating} category (thirdness) corresponds to {\ttfamily FOLE}  
list maps
between entities and attributes.
The triads {\itshape actuality}-{\itshape prehension}-{\itshape nexus} 
and {\itshape form}-{\itshape proposition}-{\itshape intention}
correspond to Whitehead's categories of existence.
%

%
\begin{figure}
\begin{center}
\begin{tabular}{c@{\hspace{60pt}}c}
\begin{tabular}{c}
\setlength{\unitlength}{0.9pt}
\begin{picture}(150,120)(0,-10)
\put(75,90){\makebox(0,0){\normalsize{$\top$}}}
\put(15,60){\makebox(0,0){\tiny{independent}}}
\put(75,60){\makebox(0,0){\tiny{mediating}}}
\put(135,60){\makebox(0,0){\tiny{relative}}}
\put(30,45){\makebox(0,0){\tiny{abstract}}}
\put(120,50){\makebox(0,0){\tiny{physical}}}

\put(0,0){\makebox(0,0){\tiny{form}}}
\put(30,5){\makebox(0,0){\tiny{actuality}}}
\put(60,0){\makebox(0,0){\tiny{intention}}}
\put(90,5){\makebox(0,0){\tiny{nexus}}}
\put(120,0){\makebox(0,0){\tiny{proposition}}}
\put(150,5){\makebox(0,0){\tiny{prehension}}}
\put(0,0){\line(2,3){30}}
\put(60,0){\line(-2,3){30}}
\put(120,0){\line(-2,1){90}}
\put(30,5){\line(2,1){90}}
\put(90,5){\line(2,3){30}}
\put(150,5){\line(-2,3){30}}
\multiput(15,60)(2,1){30}{\tiny{$\cdot$}}
\multiput(30,45)(1.5,1.5){30}{\tiny{$\cdot$}}
\multiput(75,60)(0,2){15}{\tiny{$\cdot$}}
\multiput(120,50)(-1.5,1.3){30}{\tiny{$\cdot$}}
\multiput(135,60)(-2,1){30}{\tiny{$\cdot$}}
\qbezier[35](0,0)(7.5,30)(15,60)
\qbezier[35](30,5)(22.5,30)(15,60)
\qbezier[35](60,0)(67.5,30)(75,60)
\qbezier[35](90,5)(82.5,27.5)(75,60)
\qbezier[35](120,0)(127.5,37.5)(135,60)
\qbezier[35](150,5)(142.5,35)(135,60)
\end{picture}
\end{tabular}
&
{{\begin{tabular}{c}
\setlength{\unitlength}{0.9pt}
\begin{picture}(160,120)(3,-15)
\put(75,90){\makebox(0,0){\normalsize{$\mathcal{M}$}}}
\put(15,68){\makebox(0,0){\shortstack{\tiny{entity}\\\tiny{classification}\\\scriptsize{$\mathcal{E}$}}}}
\put(75,68){\makebox(0,0){\shortstack{\tiny{designation}\\
\footnotesize{$\stackrel{{\langle{\sigma,\tau}\rangle}}{\rightrightarrows}$}}}}
\put(135,67.5){\makebox(0,0){\shortstack{\tiny{attribute}
\mbox{ }\tiny{list}\\\tiny{classification}\\\scriptsize{$\mathrmbf{List}(\mathcal{A})$}}}}
\put(30,45){\makebox(0,0){\shortstack{\scriptsize{$\mathcal{S}$}\\\tiny{schema}\mbox{ }\tiny{hypergraph}}}}
\put(120,50){\makebox(0,0){\shortstack{\scriptsize{$\mathcal{U}$}\\\tiny{universe}\mbox{ }\tiny{hypergraph}}}}
\put(0,-7){\makebox(0,0){\shortstack{\tiny{entity}\mbox{ }\tiny{type}\\\tiny{set}\\\scriptsize{$R$}}}}
\put(60,-5.3){\makebox(0,0){\shortstack{\tiny{signature}\\\tiny{map}\\\scriptsize{$\xrightarrow{\;\sigma\;}$}}}}
\put(120,-9){\makebox(0,0){\shortstack{\tiny{sort}\mbox{ }\tiny{list}\\\tiny{set}\\\scriptsize{$\mathrmbf{List}(X)$}}}}
\put(30,2.5){\makebox(0,0){\shortstack{\tiny{key}\\\tiny{set}\\\scriptsize{$K$}}}}
\put(90,4.3){\makebox(0,0){\shortstack{\tiny{tuple}\\\tiny{map}\\\scriptsize{$\xrightarrow{\;\tau\;}$}}}}
\put(151,2){\makebox(0,0){\shortstack{\tiny{value}\mbox{ }\tiny{list}\\\tiny{set}\\\scriptsize{$\mathrmbf{List}(Y)$}}}}
\dottedline{2}(15,60)(135,60)
\dottedline{2}(0,-15)(150,-15)
\dottedline{2}(30,-5)(150,-5)
\put(0,0){\line(2,3){30}}
\put(60,0){\line(-2,3){30}}
\put(120,0){\line(-2,1){90}}
\put(30,5){\line(2,1){90}}
\put(90,5){\line(2,3){30}}
\put(150,5){\line(-2,3){30}}
\multiput(15,60)(2,1){30}{\tiny{$\cdot$}}
\multiput(30,45)(1.5,1.5){30}{\tiny{$\cdot$}}
\multiput(75,60)(0,2){15}{\tiny{$\cdot$}}
\multiput(120,50)(-1.5,1.3){30}{\tiny{$\cdot$}}
\multiput(135,60)(-2,1){30}{\tiny{$\cdot$}}
\qbezier[35](0,0)(7.5,30)(15,60)
\qbezier[35](30,5)(22.5,30)(15,60)
\qbezier[35](60,0)(67.5,30)(75,60)
\qbezier[35](90,5)(82.5,27.5)(75,60)
\qbezier[35](120,0)(127.5,37.5)(135,60)
\qbezier[35](150,5)(142.5,35)(135,60)
\end{picture}
\end{tabular}}}
\\ & \\ 
{\footnotesize{hierarchy of top-level categories}}
&
{\footnotesize{{\ttfamily FOLE} structure components}}
\\ & \\ 
\end{tabular}
\end{center}
\caption{Analogy}
\label{analogy}
\end{figure}
\begin{table}
\begin{center}
{\fbox{\scriptsize\itshape{\setlength{\extrarowheight}{4pt}{\begin{tabular}{|@{\hspace{5pt}}l@{\hspace{5pt}}||@{\hspace{5pt}}c@{\hspace{5pt}}|@{\hspace{5pt}}c@{\hspace{5pt}}|@{\hspace{5pt}}c@{\hspace{5pt}}|}
\hline
& independent & mediating & relative 
\\
& (entity classification)
& (list designation)
& (type domain)
\\\hline\hline
abstract (schema) & form (entity type) & intension (signature)
 & proposition (sort)
\\\hline
physical (universe) & actuality (entity) & nexus (tuple) & prehension (value)
\\\hline
\end{tabular}}}}}
\end{center}
\caption{Matrix of six central categories}
\label{matrix}
\end{table}

Here we explain this in more detail.
As discussed in \S~\ref{sub:sub:sec:data:model:struc},
a {\ttfamily FOLE} structure is a two-dimensional combination of a classification and a hypergraph. 
The entire hierarchy of
the top-level ontological categories
is represented by 
a {\ttfamily FOLE} (model-theoretic) structure
$\mathcal{M} = {\langle{\mathcal{E},{\langle{\sigma,\tau}\rangle},\mathcal{A}}\rangle}$.
This is appropriate,
since a (model-theoretic) structure represents the knowledge in the local world of a community of discourse.
The {\itshape form}-{\itshape proposition}-{\itshape intention} triad
is represented by 
a {\ttfamily FOLE} schema
$\mathcal{S} = {\langle{R,\sigma,X}\rangle}$
with signature function
$R\xrightarrow{\sigma}\mathrmbf{List}(X)$.
The {\itshape actuality}-{\itshape prehension}-{\itshape nexus} triad
is represented by 
a {\ttfamily FOLE} universe
$\mathcal{U} = {\langle{K,\tau,Y}\rangle}$
with tuple function
$K\xrightarrow{\tau}\mathrmbf{List}(Y)$.
The firstness subgraph of {\itshape independent}({\itshape actuality},{\itshape form})
is represented by 
a {\ttfamily FOLE} entity classification
$\mathcal{E} = {\langle{R,K,\models_{\mathcal{E}}}\rangle}$
between entity instances (keys) and entity types
(or a classification
between keys and logical formula, more generally; see Kent~\cite{kent:fole:era:supstruc}).
The secondness subgraph of {\itshape relative}({\itshape prehension},{\itshape proposition})
is represented by 
a {\ttfamily FOLE} attribute classification 
$\mathcal{A} = {\langle{X,Y,\models_{\mathcal{A}}}\rangle}$
between attribute instances (data values) and attribute types (sorts).
The thirdness subgraph of {\itshape mediating}({\itshape nexus},{\itshape intention})
is represented by
a {\ttfamily FOLE} list designation 
${\langle{\sigma,\tau}\rangle} : \mathcal{E} \rightrightarrows \mathrmbf{List}(\mathcal{A})$.
%
A {\ttfamily FOLE} entity type loosely corresponds to a 
$\text{\itshape form} = \text{\itshape independent}{\,\wedge\,}\text{\itshape abstract}$.
A {\ttfamily FOLE} entity (key) loosely corresponds to an 
$\text{\itshape actuality} = \text{\itshape independent}{\,\wedge\,}\text{\itshape physical}$.
A {\ttfamily FOLE} sort (attribute type),
which is not necessariy linked from a {\ttfamily FOLE} entity, 
loosely	corresponds to a 
$\text{\itshape proposition} = \text{\itshape relative}{\,\wedge\,}\text{\itshape abstract}$.
A {\ttfamily FOLE} data value (attribute instance) loosely corresponds to a 
$\text{\itshape prehension} = \text{\itshape relative}{\,\wedge\,}\text{\itshape physical}$.
\subsection{Linearization.}\label{sub:sub:sec:misc:linear}

The ``ontology log'' {\ttfamily Olog} formalism 
(Spivak and Kent~\cite{spivak:kent:olog})
is a category-theoretic model for knowledge representation.
As we indicate below,
there is a sense in which
the {\ttfamily FOLE} representation subsumes the {\ttfamily Olog} representation
(and vice-versa).
\footnote{
This section is closely related to 
the discussion of the sketch and interpretation associated with a unified relational database in \S~4 of Kent~\cite{kent:db:sem}.
}
Let $\mathcal{M}$ be a {\ttfamily FOLE} structure with 
schema $\mathcal{S} = {\langle{R,\sigma,X}\rangle}$,
universe $\mathcal{U} = {\langle{K,\tau,Y}\rangle}$,
type domain $\mathcal{A} = {\langle{X,Y,\models_{\mathcal{A}}}\rangle}$,
entity classification $\mathcal{E} = {\langle{R,K,\models_{\mathcal{E}}}\rangle}$, and
list designation ${\langle{\sigma,\tau}\rangle} : \mathcal{E} \rightrightarrows \mathrmbf{List}(\mathcal{A})$.

\newpage

By using the extent operator
(\S~\ref{sub:sub:sec:data:model:struc}),
the following classifications and sets are informationally equivalent.
\begin{center}
{\scriptsize{\setlength{\extrarowheight}{2pt}$\begin{array}{|l@{\hspace{8pt}}|l@{\hspace{3pt}$=$\hspace{3pt}}l|}
\multicolumn{1}{l|}{\text{\itshape{classification}}\;\;}
& \multicolumn{2}{l}{\,\text{\itshape{set}}}
\\\hline
\,\mathcal{E} 
& \,\coprod\mathrmbfit{ext}_{\mathcal{E}}
& \{ (r,k) \mid r \in R, k \in K, k \models_{\mathcal{E}} r \}
\\
\,\mathcal{A} 
& \,\coprod\mathrmbfit{ext}_{\mathcal{A}}
& \{ (x,y) \mid x \in X, y \in Y, y \models_{\mathcal{A}} x \}
\\
\,\mathrmbf{List}(\mathcal{A}) 
& \,\coprod\mathrmbfit{ext}_{\mathrmbf{List}(\mathcal{A})}
& \{ (I,s,t) \mid (I,s) \in \mathrmbf{List}(X), (I,t) \in \mathrmbf{List}(Y), (I,t) \models_{\mathrmbf{List}(\mathcal{A})} (I,s) \}
\\\hline
\end{array}$}}
\end{center}
The designation property defines a function
$\coprod\mathrmbfit{ext}_{\mathcal{E}}
\xrightarrow{\coprod\mathrmbfit{ext}_{{\langle{\sigma,\tau}\rangle}}}
\coprod\mathrmbfit{ext}_{\mathrmbf{List}(\mathcal{A})}$,
which is the sum of the
extent diagram morphism 
$\mathrmbfit{ext}_{{\langle{\sigma,\tau}\rangle}}
= {\langle{\sigma,\tau_{{\scriptscriptstyle{\llcorner}}}}\rangle}
: {\langle{R,\mathrmbfit{ext}_{\mathcal{E}}}\rangle}
\rightarrow
{\langle{\mathrmbf{List}(X),\mathrmbfit{ext}_{\mathrmbf{List}(\mathcal{A})}}\rangle}$
(Fig.~\ref{fig:interp:struc:obj} of \S~\ref{sub:sub:sec:data:model:struc}).
By flattening the resulting lists,
for each $(r,k){\,\in\,}\coprod\mathrmbf{ext}_{\mathcal{E}}$,
we know that
$(s_{i},t_{i}){\,\in\,}\coprod\mathrmbfit{ext}_{\mathcal{A}}$ 
for each index $i\in\alpha(r)=I$ in the common arity.
%
%
Hence,
the following linearization set
\footnote{The linearization set $\mathrmbfit{lin}(\mathcal{M})$ loosely corresponds to the following:
(1) 
the table used in the entity-attribute-value ({\ttfamily EAV}) data model,
where data is recorded in three columns:
the entity component is a foreign key into an object information table,
the attribute component is a foreign key into an attribute definition table, and
the value component is the value of the attribute; and
(2)
the labelled directed graph used in the resource description framework ({\ttfamily RDF}) data model,
where each subject-predicate-object triple is regarded as an edge in the graph.}
is equivalent in information to the structure $\mathcal{M}$:
\[\mbox{\footnotesize{$\mathrmbfit{lin}(\mathcal{M})=
\Bigl\{ 
{\langle{\underset{\scriptstyle{\text{ent}}}{\underbrace{(r,k)}},\underset{\scriptstyle{\text{attr}}}{\underbrace{i,(s_{i},t_{i})}}}\rangle} 
\mid 
(r,k){\,\in\,}\coprod\mathrmbf{ext}_{\mathcal{E}},
\sigma(r) = {\langle{I,s}\rangle},
\tau(k) = {\langle{I,t}\rangle},
i{\,\in\,}I
\Bigr\}
$.}\normalsize}\]
The {\ttfamily FOLE} representation is binary,
since it has two kinds of type, sorts (attribute types) and entity types.
The {\ttfamily Olog} representation is unary, 
since it has only one kind of type, the abstract concept.
\footnote{See \S~4.4 of \cite{spivak:kent:olog}
for further discussion of binary/unary knowledge representations.}
However, 
the {\ttfamily FOLE} representation can be transformed to the {\ttfamily Olog} representation by the process of \emph{linearization}.
If we restrict {\ttfamily FOLE} to the unified model, 
identifying entities with attributes $\mathcal{E}{\;=\;}\mathcal{A}$,
by separating the type/instance information in 
the linearization set 
$\mathrmbfit{lin}(\mathcal{M})$
via the extent operator 
along the classification dimension (Fig.~\ref{fig:fole:data:model}),
we get the basis for the {\ttfamily Olog} data model,
consisting of three notions: types, aspects, and facts.
Types, which represent things, are depicted by nodes in {\ttfamily Olog};
aspects, which represent functional relationships between things, are depicted by edges in {\ttfamily Olog}; and
facts, which represent assertions, are depicted by path equations in {\ttfamily Olog}.

{\ttfamily Olog} types and aspects are covered by the linearization process discussed here.
{\ttfamily Olog} facts correspond to the formalism 
discussed (indirectly) in the {\ttfamily FOLE} superstructure (Kent~\cite{kent:fole:era:supstruc})
and (more directly) in the full form of {\ttfamily FOLE} (Kent~\cite{kent:iccs2013}).
\begin{center}
{{\begin{tabular}{c}
\setlength{\unitlength}{0.55pt}\begin{picture}(280,130)(-30,-30)
\put(20,80){\makebox(0,0){\footnotesize{$\underset{\scriptscriptstyle{\text{typ}}}{r}
\xrightarrow{\;(r,i)\;}
\underset{\scriptscriptstyle{\text{typ}}}{s_{i}}$}}}
\put(110,80){\makebox(0,0)[l]{\footnotesize{$\left.\rule{0pt}{10pt}\right\}$}}}
\put(150,80){\makebox(0,0){\footnotesize{$\mathrmbf{S}$}}}
\put(190,80){\makebox(0,0)[l]{\footnotesize{a schema context}}}
\thicklines
\put(20,55){\vector(0,-1){30}}
\put(40,40){\makebox(0,0){\scriptsize{$\mathrmbfit{ext}$}}}
\put(150,60){\vector(0,-1){40}}
\put(158,40){\makebox(0,0)[l]{\footnotesize{$\mathrmbfit{M}$}}}
\put(190,40){\makebox(0,0)[l]{\footnotesize{an interpretation passage}}}
\put(20,0){\makebox(0,0){\scriptsize{$
\underset{\mathrmbfit{ext}(r)}{\mathrmbfit{M}(r)}
\xrightarrow[\tau_{r}{\,\cdot\,}\mathrmbfit{tup}(\pi_{i})]{\;\;\mathrmbfit{M}(i)\;\;}
\underset{\mathrmbfit{ext}(s_{i})}{\mathrmbfit{M}(s_{i})}$}}}
\put(110,0){\makebox(0,0)[l]{\footnotesize{$\left.\rule{0pt}{10pt}\right\}$}}}
\put(150,0){\makebox(0,0){\footnotesize{$\mathrmbf{Set}$}}}
\put(190,0){\makebox(0,0)[l]{\footnotesize{the context of sets}}}
\put(20,-30){\makebox(0,0){\scriptsize{$k\xmapsto{\;\;\;\;}t_{i}$}}}
\end{picture}
\end{tabular}}}
\end{center}
%

\newpage

\begin{description}
\item[{\ttfamily Olog} Database Schema $\mathrmbf{S}$:] 
By projecting the type components out of the linearization set $\mathrmbfit{lin}(\mathcal{M})$,
define a graph $\mathcal{G}$
whose node set is $R=X$ and whose edge set is
$\coprod_{r\in{R=X}}\alpha(r)$.
Picture the graph $\mathcal{G}$ as follows: 
\[\mbox{\footnotesize{$
\mathcal{G} = 
\Bigl\{ r\xrightarrow{(r,i)}s_{i}
\;\Big|\; 
\sigma(r) = {\langle{I,s}\rangle},\;
i{\,\in\,}I,
\text{ with }
r,s_{i} \in R=X
\Bigr\}\;.
$}\normalsize}\]
The schema mathematical context
$\mathrmbf{S}=\mathcal{G}^{\ast}$ is the path context of graph $\mathcal{G}$.
The graph $\mathcal{G}$ corresponds to a many-sorted unary algebraic language $\mathcal{O}={\langle{X,\Omega}\rangle}$, 
where $\Omega$ is essentially the opposite of $\mathcal{G}$:
the collection of sets of function (operator) symbols
$\Omega = 
\{ \Omega_{s_{i},{\langle{1,r}\rangle}} 
\mid \text{sort}\;s_{i} \in X, \text{unary signature}\;{\langle{1,r}\rangle} \in \mathrmbf{List}(X) \}$
with each $(r,i) \in \Omega_{s_{i},{\langle{1,r}\rangle}}$
being an $s_{i}$-sorted $r={\langle{1,r}\rangle}$-ary function symbol
$s_{i} \xrightharpoondown{(r,i)} r={\langle{1,r}\rangle}$.
\footnote{Define ${\langle{X,\Omega}\rangle}$-term
${\langle{I,s}\rangle} \xrightharpoondown{r} {\langle{1,r}\rangle} = r$
as the cotupling of
$\{ s_{i} \xrightharpoondown{(r,i)} r={\langle{1,r}\rangle} \mid i{\,\in\,}I \}$,
so that
$s_{i} \xrightharpoondown{(r,i)} r={\langle{1,r}\rangle}$
is the composite
$s_{i} \xrightharpoondown{\pi_{i}} {\langle{I,s}\rangle} \xrightharpoondown{r} {\langle{1,r}\rangle} = r$.}
The opposite of the schema context is a subcontext (no cotupling) of the term context
$\mathrmbf{Term}_{\mathcal{O}}\;\xhookleftarrow{\mathrmbfit{inc}}\;\mathrmbf{S}^{\text{op}}$.
An {\ttfamily Olog} fact corresponds to 
a linear $\mathcal{O}$-equation
$(t{\,=\,}\hat{t}) : {\langle{1,s}\rangle}{\,\rightharpoondown\,}{\langle{1,s'}\rangle}$.
\newline
\item[{\ttfamily Olog} Database Instance $\mathrmbfit{M}$:] 
Using extent,
define a passage $\mathrmbf{S}=\mathcal{G}^{\ast}\xrightarrow{\mathrmbfit{M}}\mathrmbf{Set}$
which maps a node $r$ in $\mathrmbf{S}$ 
to the extent $\mathrmbfit{M}(r) = K(r) = \mathrmbfit{ext}(r)$
and maps an edge $r\xrightarrow{(r,i)}s_{i}$ in $\mathrmbf{S}$
to the function
$\mathrmbfit{M}(r)\xrightarrow[\tau_{r}{\,\cdot\,}\mathrmbfit{tup}(\pi_{i})]{\mathrmbfit{M}(r,i)}\mathrmbfit{M}(s_{i})$.
\footnote{A projection
{\footnotesize{$
\mathrmbfit{M}(r)=\mathrmbfit{ext}(r)
\xrightarrow{\tau_{r}}\mathrmbfit{tup}(I,s)
\xrightarrow{\mathrmbfit{tup}(\pi_{i})}\mathrmbfit{tup}(1,s_{i})=\mathrmbfit{ext}(s_{i})=\mathrmbfit{M}(s_{i})
$}\normalsize}
of the tabular interpretation in \S~\ref{sub:sub:sec:data:model:interp}.}
%
%
%
The functional language $\mathcal{O}={\langle{X,\Omega}\rangle}$ mentioned above has 
an $\mathcal{O}$-algebra ${\langle{A,\delta}\rangle}$
consisting of the $X$-sorted collection of extents
$\{ A_{x} = \mathrmbfit{ext}(x) \subseteq Y \mid x \in X \}$,
where $\delta$ assigns 
the ${\langle{1,r}\rangle}$-ary $s_{i}$-sorted function 
$\mathrmbfit{M}(r)=A_{r}=A^{{\langle{1,r}\rangle}}
\xrightarrow[\mathrmbfit{M}(r,i)]{\delta_{(r,i)}}
A_{s_{i}}=\mathrmbfit{M}(s_{i})$
to each function symbol $s_{i} \xrightharpoondown{(r,i)} r={\langle{1,r}\rangle}$.
\footnote{$\delta$ assigns 
the table
$\mathrmbfit{M}(r)=\mathrmbfit{ext}(r)
\xrightarrow[\tau_{r}]{\delta_{r}}
\mathrmbfit{tup}(I,s)=\prod_{i{\,\in\,}I}A_{s_{i}}=\prod_{i{\,\in\,}I}\mathrmbfit{M}(s_{i})$
(a tupling)
to the cotupling
${\langle{I,s}\rangle} \xrightharpoondown{r} r$.}
%
The interpretation passage 
$\mathrmbf{Term}_{\mathcal{O}}^{\mathrm{op}}\xrightarrow{\mathrmbfit{A}}\;\mathrmbf{Set}$
inductively define by this $\mathcal{O}$-algebra 
contains the database instance as a subfunctor
$\mathrmbfit{M}=\mathrmbfit{inc}^{\mathrm{op}}{\circ\;}\mathrmbfit{A}$.
Any equation satisfied by the {\ttfamily Olog} database interpretation $\mathrmbfit{M}$
is also satisfied by the interpretation passage $\mathrmbfit{A}$.
\footnote{The interpretation passage satisfies an $\mathcal{O}$-equation
$(t{\,=\,}\hat{t}) : {\langle{I,s}\rangle}{\,\rightharpoondown\,}{\langle{I',s'}\rangle}$
when the operations coincide
$\mathrmbfit{A}(t) = \mathrmbfit{A}(\hat{t}) :
\mathrmbfit{A}(I,s)\leftarrow\mathrmbfit{A}(I',s')$ in $\mathrmbf{Set}$ (see Kent~\cite{kent:iccs2013}).}
\end{description}


\section{Conclusion}

The {\ttfamily ERA} data model
describes the conceptual model for a community of discourse,
which can be used as the foundation for designing relational databases.
The first-order logical environment {\ttfamily FOLE}
provides a rigorous and principled approach to distributed inter-operable first-order information systems, 
which integrates ontologies and databases into a unified framework.
In this paper, 
we have discussed in detail the representation of elements of the {\ttfamily ERA} data model 
with components of the {\ttfamily FOLE} logical environment,
described interpretation basics,
and illustrated the connections between {\ttfamily FOLE} and other data models in knowledge representation.
%
We 
complete the presentation of the {\ttfamily FOLE} logical environment 
by discussing
the formalism and semantics of many-sorted logic in the paper Kent~\cite{kent:fole:era:supstruc} 
and
database interpretation in the papers
Kent~\cite{kent:fole:era:tbl}
and
Kent~\cite{kent:fole:era:db}.



\appendix
\section{The Fibered Context of Structures.}\label{sub:sec:struc:fbr:pass}

In order to allow communities of discourse to interoperate,
we define the notion of a morphism between two structures 
that respects the {\ttfamily ERA} data model.
This describes
the mathematical context of structures
$\mathrmbf{Struc}$ as a \emph{fibered} mathematical context 
in two orientations: 
the Grothedieck construction of 
the schema indexed mathematical context 
$\mathrmbf{Sch}^{\mathrm{op}}{\!\xrightarrow{\;\mathrmbfit{struc}^{\curlywedge}}\,}\mathrmbf{Cxt}$ or 
the Grothedieck construction of 
the universe indexed mathematical context 
$\mathrmbf{Univ}^{\mathrm{op}}{\!\xrightarrow{\;\mathrmbfit{struc}^{\curlyvee}}\,}\mathrmbf{Cxt}$.
The schema indexed mathematical context of structures 
is used in 
the paper \cite{kent:fole:era:supstruc} on {\ttfamily FOLE} superstructure
to establish the institutional aspect of the {\ttfamily FOLE}.
Each orientation is developed in three steps:
the structure fiber context for a single indexing object,
the structure fiber passage along an indexing morphism,
the fibered mathematical context $\mathrmbf{Struc}$ of structures and structure morphisms
as the Grothedieck construction for this orientation.
This dual development of
the mathematical context of structures
$\mathrmbf{Struc}$ 
(Prop.~\ref{prop:struc:mor:dual:facts})
is based upon and mirrors 
a dual development of
the mathematical context of classifications
$\mathrmbf{Cls}$ 
(Prop.~\ref{prop:cls:mor:dual:facts}).

\newpage
\subsection{An Exemplar.}\label{sub:sub:sec:cls:fbr:cxt}

The mathematical context of classifications
$\mathrmbf{Cls}$ is a \emph{fibered} mathematical context 
in two orientations: 
the Grothedieck construction of 
the type set indexed mathematical context 
$\mathrmbf{Set}^{\mathrm{op}}{\!\xrightarrow{\;\mathrmbfit{cls}^{\curlywedge}}\,}\mathrmbf{Cxt}$ or 
the Grothedieck construction of 
the instance set indexed mathematical context 
$\mathrmbf{Set}^{\mathrm{op}}{\!\xrightarrow{\;\mathrmbfit{cls}^{\curlyvee}}\,}\mathrmbf{Cxt}$.
\begin{proposition}\label{prop:cls:mor:dual:facts}
Any infomorphism 
$\mathcal{C}_{2}={\langle{X_{2},Y_{2},\models_{2}}\rangle}
\xrightleftharpoons{{\langle{f,g}\rangle}}
{\langle{X_{1},Y_{1},\models_{1}}\rangle}=\mathcal{C}_{1}$
in the mathematical context of classifications $\mathrmbf{Cls}$,
with type function projection $X_{2}\xrightarrow{f}X_{1}$
and instance function projection $Y_{2}\xleftarrow{g}Y_{1}$,
has dual factorizations
\begin{center}
{{\begin{tabular}{c}
\setlength{\unitlength}{0.6pt}
\begin{picture}(120,140)(0,-40)
\put(60,90){\makebox(0,0){\footnotesize{$\underset{\textstyle{\overbrace{\rule{80pt}{0pt}}}}{\mathrmbf{Cls}(X_{2})}$}}}
\put(60,-30){\makebox(0,0){\footnotesize{$\overset{\textstyle{\underbrace{\rule{80pt}{0pt}}}}{\mathrmbf{Cls}(Y_{1})}$}}}
\put(0,60){\makebox(0,0){\footnotesize{$\mathcal{C}_{2}$}}}
\put(92,60){\makebox(0,0)[l]{\footnotesize{$\mathrmbfit{cls}^{\curlywedge}_{f}(\mathcal{C}_{1})={\langle{X_{2},Y_{1},\models_{f}}\rangle}$}}}
\put(120,0){\makebox(0,0){\footnotesize{$\mathcal{C}_{1}$}}}
\put(28,0){\makebox(0,0)[r]{\footnotesize{${\langle{X_{2},Y_{1},\models_{g}}\rangle}=\mathrmbfit{cls}^{\curlyvee}_{g}(\mathcal{C}_{2})$}}}
\put(-7,30){\makebox(0,0)[r]{\scriptsize{$g$}}}
\put(7,30){\makebox(0,0)[l]{\tiny{${\langle{1_{X_{2}},g}\rangle}$}}}
\put(73,10){\makebox(0,0){\scriptsize{$f$}}}
\put(73,-12){\makebox(0,0){\tiny{${\langle{f,1_{Y_{1}}}\rangle}$}}}
\put(48,70){\makebox(0,0){\scriptsize{$g$}}}
\put(48,48){\makebox(0,0){\tiny{${\langle{1_{X_{2}},g}\rangle}$}}}
\put(113,30){\makebox(0,0)[r]{\scriptsize{$f$}}}
\put(127,30){\makebox(0,0)[l]{\tiny{${\langle{f,1_{Y_{1}}}\rangle}$}}}
%
%
\put(48,57){\makebox(0,0){\large{$\xrightleftharpoons{\;\;\;\;\;\;\;\;}$}}}
\put(72,-3){\makebox(0,0){\large{$\xrightleftharpoons{\;\;\;\;\;\;\;\;}$}}}
\put(118,30){\makebox(0,0){\large{$\upharpoonleft$}}}
\put(122,30){\makebox(0,0){\large{$\downharpoonright$}}}
\put(-2,30){\makebox(0,0){\large{$\upharpoonleft$}}}
\put(2,30){\makebox(0,0){\large{$\downharpoonright$}}}
\end{picture}
\end{tabular}}}
\end{center}
where 
$\mathrmbfit{cls}^{\curlywedge}_{f}(\mathcal{C}_{1})
= f^{-1}(\mathcal{C}_{1})
= {\langle{X_{2},Y_{1},\models_{f}}\rangle}
= {\langle{X_{2},Y_{1},\models_{g}}\rangle}
= g^{-1}(\mathcal{C}_{2})
= \mathrmbfit{cls}^{\curlyvee}_{g}(\mathcal{C}_{2})$.
\end{proposition}
\begin{proof}
$y_{1}{\,\models_{g}\,}x_{2}$
\underline{iff}
$g(y_{1}){\,\models_{\mathcal{C}_{2}}\,}x_{2}$
\underline{iff}
$y_{1}{\,\models_{\mathcal{C}_{1}}\,}f(x_{2})$
\underline{iff}
$y_{1}{\,\models_{f}\,}x_{2}$.
\end{proof}
The top-right factorization 
consists of 
a $\mathrmbf{Cls}(X_{2})$-morphism 
$\mathcal{C}_{2}\xrightleftharpoons[{\langle{1_{X_{2}},g}\rangle}]{g}\mathrmbfit{cls}^{\curlywedge}_{f}(\mathcal{C}_{1})$
and 
the $\mathcal{C}_{1}^{\text{th}}$ component 
$\mathrmbfit{cls}^{\curlywedge}_{f}(\mathcal{C}_{1})\xrightleftharpoons[{\langle{f,1_{Y_{1}}}\rangle}]{f}\mathcal{C}_{1}$
of a bridge
$\mathrmbfit{cls}^{\curlywedge}_{f}{\;\circ\;}\mathrmbfit{inc}_{X_{2}}
\xRightarrow{\,\grave{\chi}_{f}\;\,}\mathrmbfit{inc}_{X_{1}}$.
This factors through 
the $X_{2}$-classification $\mathrmbfit{cls}^{\curlywedge}_{f}(\mathcal{C}_{1})$,
which is the fiber passage image 
along the type function $X_{2}\xrightarrow{f}X_{1}$
of the the $X_{1}$-classification $\mathcal{C}_{1}$.
The left-bottom factorization consists of 
the $\mathcal{C}_{2}^{\text{th}}$ component 
$\mathcal{C}_{2}\xrightleftharpoons[{\langle{1_{X_{2}},g}\rangle}]{g}\mathrmbfit{cls}^{\curlyvee}_{g}(\mathcal{C}_{2})$
of a bridge
$\mathrmbfit{cls}^{\curlyvee}_{g}{\;\circ\;}\mathrmbfit{inc}_{Y_{1}}
\xLeftarrow{\;\,\acute{\chi}_{g}}\mathrmbfit{inc}_{Y_{2}}$
and
a $\mathrmbf{Cls}(Y_{1})$-morphism
$\mathrmbfit{cls}^{\curlyvee}_{g}(\mathcal{C}_{2})\xrightleftharpoons[{\langle{f,1_{Y_{1}}}\rangle}]{f}\mathcal{C}_{1}$.
This factors through 
the $Y_{1}$-classification $\mathrmbfit{cls}^{\curlyvee}_{g}(\mathcal{C}_{2})$,
which is the fiber passage image 
along the instance function $Y_{2}\xleftarrow{g}Y_{1}$
of the the $Y_{2}$-classification $\mathcal{C}_{2}$.
%

\newpage
\subsection{Schema Orientation.}\label{sub:sub:sec:sch:index}

The fibered mathematical context $\mathrmbf{Struc}$ of structures and structure morphisms
can be developed as
the Grothedieck construction of 
a schema indexed mathematical context.
This approach using schema indexing corresponds to the use of type-set indexing in \S~\ref{sub:sub:sec:cls:fbr:cxt},
and was the approach used in 
the paper 
(Kent \cite{kent:iccs2013}).
\begin{enumerate}
\item 
For a fixed schema $\mathcal{S}$,
we define the structure fiber context $\mathrmbf{Struc}(\mathcal{S})$.
\item 
For a schema morphism
$\mathcal{S}_{2}
\xRightarrow{{\langle{r,f}\rangle}}
\mathcal{S}_{1}$,
we define the structure fiber passage
$\mathrmbf{Struc}(\mathcal{S}_{2})\xleftarrow{\mathrmbfit{struc}^{\curlywedge}_{{\langle{r,f}\rangle}}}\mathrmbf{Struc}(\mathcal{S}_{1})$.
\rule[-8pt]{0pt}{10pt}
\item 
We define 
the fibered mathematical context $\mathrmbf{Struc}$ of structures and structure morphisms
to be the Grothedieck construction of 
the schema indexed mathematical context 
$\mathrmbf{Sch}^{\mathrm{op}}{\xrightarrow{\mathrmbfit{struc}^{\curlywedge}}}\;\mathrmbf{Cxt}$.
\end{enumerate}
%


\paragraph{Fixed Schema.}
%
Let
$\mathcal{S} = {\langle{R,\sigma,X}\rangle}\in\mathrmbf{Sch}$
be a fixed schema.
We define the notion of a morphism between two (fixed schema) $\mathcal{S}$-structures 
that respects the {\ttfamily ERA} data model.
In these morphisms,
the schema remains fixed,
but the attribute 
and entity instances (data values and keys)
are formally linked by maps that respect 
universe tuple, typed domain extent and entity interpretation. 
An $\mathcal{S}$-structure morphism
$\mathcal{M}_{2}\xrightleftharpoons{{\langle{k,g}\rangle}}\mathcal{M}_{1}$
over a fixed schema $\mathcal{S} = {\langle{R,\sigma,X}\rangle}$
is
a universe morphism  
$\mathrmbfit{univ}(\mathcal{M}_{2})
\xLeftarrow{\;{\langle{k,g}\rangle}\;}
\mathrmbfit{univ}(\mathcal{M}_{1})$,
where
$\mathrmbfit{attr}(\mathcal{M}_{2})\xrightleftharpoons{{\langle{1_{X},g}\rangle}}\mathrmbfit{attr}(\mathcal{M}_{1})$
is an infomorphism in $\mathrmbf{Cls}(X)$ over the sort set $X$, and
$\mathrmbfit{ent}(\mathcal{M}_{2})\xrightleftharpoons{{\langle{1_{R},k}\rangle}}\mathrmbfit{ent}(\mathcal{M}_{1})$
is an infomorphism in $\mathrmbf{Cls}(R)$ over the entity type set $R$.
\footnote{Hence,
the target type domain is the inverse image of the source type domain
$\mathrmbfit{attr}(\mathcal{M}_{1}) 
= g^{-1}(\mathrmbfit{attr}(\mathcal{M}_{2}))
= \mathrmbfit{cls}^{\curlyvee}_{g}(\mathrmbfit{attr}(\mathcal{M}_{2}))$
and
the target entity classification is the inverse image of the source entity classification
$\mathrmbfit{ent}(\mathcal{M}_{1}) 
= k^{-1}(\mathrmbfit{ent}(\mathcal{M}_{2}))
= \mathrmbfit{cls}^{\curlyvee}_{k}(\mathrmbfit{ent}(\mathcal{M}_{2}))$.
By combining entity/attribute inverse image classifications,
the target structure $\mathcal{M}_{1}$ is the inverse image 
$\mathrmbfit{struc}^{\curlyvee}_{{\langle{k,g}\rangle}}(\mathcal{M}_{2}) 
={\langle{k^{-1}(\mathrmbfit{ent}(\mathcal{M}_{2})),{\langle{\sigma,\tau_{1}}\rangle},g^{-1}(\mathrmbfit{attr}(\mathcal{M}_{2}))}\rangle}$
of the source structure $\mathcal{M}_{2}$
(\S~\ref{sub:sub:sec:univ:index}).
The two are bridged 
$\mathcal{M}_{2} 
\xrightleftharpoons{{\langle{k,g}\rangle}}
\mathrmbfit{struc}^{\curlyvee}_{{\langle{k,g}\rangle}}(\mathcal{M}_{2})=\mathcal{M}_{1}$
by the structure morphism.}
$\mathcal{S}$-structure morphisms compose component-wise.
Let $\mathrmbf{Struc}(\mathcal{S})$ 
denote the fiber context of structures
over the fixed schema $\mathcal{S}$.


\paragraph{Structure Fiber Passage.}

We define 
the indexed 
context 
$\mathrmbf{Sch}^{\mathrm{op}}{\xrightarrow{\mathrmbfit{struc}^{\curlywedge}}\!}\;\mathrmbf{Cxt}$.
%
%
Given a schema $\mathcal{S}$,
there is a fiber context of structures $\mathrmbf{Struc}(\mathcal{S})$ with that schema.
Given a schema morphism
$\mathcal{S}_{2}\xRightarrow{{\langle{r,f}\rangle}}\mathcal{S}_{1}$,
there is a fiber passage of structures
$\mathrmbf{Struc}(\mathcal{S}_{2})
\xleftarrow{\mathrmbfit{struc}^{\curlywedge}_{{\langle{r,f}\rangle}}}
\mathrmbf{Struc}(\mathcal{S}_{1})$:
a structure 
$\mathcal{M}_{1}={\langle{\mathcal{E}_{1},{\langle{\sigma_{1},\tau_{1}}\rangle},\mathcal{A}_{1}}\rangle}\in\mathrmbf{Struc}(\mathcal{S}_{1})$
is mapped to a structure 
$\mathcal{M}_{2}=\mathrmbfit{struc}^{\curlywedge}_{{\langle{r,f}\rangle}}(\mathcal{M}_{1}) 
={\langle{r^{-1}(\mathcal{E}_{1}),{\langle{\sigma_{2},\tau_{1}}\rangle},f^{-1}(\mathcal{A}_{1})}\rangle}
\in\mathrmbf{Struc}(\mathcal{S}_{2})$.
$\mathcal{M}_{2}$ is called the reduct of $\mathcal{M}_{1}$ and
$\mathcal{M}_{1}$ is called the expansion of $\mathcal{M}_{2}$.
The two are linked 
$\mathrmbfit{struc}^{\curlywedge}_{{\langle{r,f}\rangle}}(\mathcal{M}_{1})
\xrightleftharpoons[{\langle{r,f}\rangle}]{\grave{\chi}_{\mathcal{M}_{1}}}
\mathcal{M}_{1}$
by a structure morphism,
which is the $\mathcal{M}_{1}^{\text{th}}$ component of a bridge
$\mathrmbfit{struc}^{\curlywedge}_{{\langle{r,f}\rangle}}{\;\circ\;}\mathrmbfit{inc}_{\mathcal{S}_{2}}
\xRightarrow{\;\grave{\chi}_{{\langle{r,f}\rangle}}\;\,}\mathrmbfit{inc}_{\mathcal{S}_{1}}$.
\begin{center}
\begin{tabular}{c@{\hspace{110pt}}c}
\begin{tabular}{c}
\setlength{\unitlength}{0.56pt}
\begin{picture}(120,120)(0,20)
\put(-7,120){\makebox(0,0){\footnotesize{$\mathrmbf{Struc}(\mathcal{S}_{2})$}}}
\put(127,120){\makebox(0,0){\footnotesize{$\mathrmbf{Struc}(\mathcal{S}_{1})$}}}
\put(60,30){\makebox(0,0){\footnotesize{$\mathrmbf{Struc}$}}}
\put(60,133){\makebox(0,0){\scriptsize{$\mathrmbfit{struc}^{\curlywedge}_{{\langle{r,f}\rangle}}$}}}
\put(20,73){\makebox(0,0)[r]{\scriptsize{$\mathrmbfit{inc}_{\mathcal{S}_{2}}$}}}
\put(100,73){\makebox(0,0)[l]{\scriptsize{$\mathrmbfit{inc}_{\mathcal{S}_{1}}$}}}
\put(60,90){\makebox(0,0){\large{$\xRightarrow{\;\grave{\chi}_{{\langle{r,f}\rangle}}}$}}}
\put(80,122){\vector(-1,0){40}}
\put(10,105){\vector(2,-3){40}}
\put(110,105){\vector(-2,-3){40}}
\end{picture}
\end{tabular}
&
{{\begin{tabular}{c}
\setlength{\unitlength}{0.58pt}
\begin{picture}(180,180)(-40,-80)
\put(0,80){\makebox(0,0){\footnotesize{$r^{-1}(\mathcal{E}_{1})$}}}
\put(46,0){\makebox(0,0)[r]{\footnotesize{${\scriptstyle\sum}_{f}^{-1}(\mathrmbf{List}(\mathcal{A}_{1}))=
\mathrmbf{List}(f^{-1}(\mathcal{A}_{1}))$}}}
\put(-15,40){\makebox(0,0)[r]{\scriptsize{${\langle{\sigma_{2},\tau_{1}}\rangle}$}}}
\put(0,40){\makebox(0,0){\large{$\downdownarrows$}}}
\put(120,80){\makebox(0,0){\footnotesize{$\mathcal{E}_{1}$}}}
\put(120,0){\makebox(0,0){\footnotesize{$\mathrmbf{List}(\mathcal{A}_{1})$}}}
\put(105,40){\makebox(0,0)[r]{\scriptsize{${\langle{\sigma_{2},\tau_{2}}\rangle}$}}}
\put(120,40){\makebox(0,0){\large{$\downdownarrows$}}}
\put(65,80){\makebox(0,0){\footnotesize{$\xrightleftharpoons{{\langle{r,1_{K_{1}}}\rangle}}$}}}
\put(65,-20){\makebox(0,0){\scriptsize{$\mathrmbf{List}(f,1_{Y_{1}})$}}}
\put(65,-3){\makebox(0,0){\normalsize{$\xrightleftharpoons{\;\;}$}}}
\put(60,-60){\makebox(0,0){\footnotesize{$
\overset{\underbrace{\rule{80pt}{0pt}}}
{\mathrmbfit{struc}^{\curlywedge}_{{\langle{r,f}\rangle}}(\mathcal{M}_{1})
\xrightleftharpoons[{\langle{r,1_{K_{1}},f,1_{Y_{1}}}\rangle}]{{(\grave{\chi}_{{\langle{r,f}\rangle}})}_{\mathcal{M}_{1}}} 
\mathcal{M}_{1}\;\;\;\;\;\;}$}}}
\end{picture}
\end{tabular}}}
\end{tabular}
\end{center}
%

\paragraph{Multiple Universes.}

The Grothedieck construction of 
the schema indexed mathematical context 
$\mathrmbf{Sch}^{\mathrm{op}}\xrightarrow{\mathrmbfit{struc}^{\curlywedge}}\mathrmbf{Cxt}$
is the fibered mathematical context $\mathrmbf{Struc}$ of structures and structure morphisms.
A structure 
$\mathcal{M}$
is as described 
in \S~\ref{sub:sub:sec:data:model:struc}
(Fig.~\ref{fig:fole:struc}).
A structure morphism
$\mathcal{M}_{2}\xrightleftharpoons{{\langle{r,k,f,g}\rangle}}\mathcal{M}_{1}$
from source structure 
$\mathcal{M}_{2}$
to target structure
$\mathcal{M}_{1}$
consists of a schema morphism 
$\mathcal{S}_{2}\xRightarrow{{\langle{r,f}\rangle}}\mathcal{S}_{1}$
from source schema 
$\mathcal{S}_{2}$
to target structure
$\mathcal{S}_{1}$
and a 
morphism 
\[\mbox{\footnotesize{$
\mathcal{M}_{2}
={\langle{\mathcal{E}_{2},{\langle{\sigma_{2},\tau_{2}}\rangle},\mathcal{A}_{2}}\rangle}
\xrightleftharpoons{{\langle{k,g}\rangle}}
{\langle{r^{-1}(\mathcal{E}_{1}),{\langle{\sigma_{2},\tau_{1}}\rangle},f^{-1}(\mathcal{A}_{1})}\rangle}=
\mathrmbfit{struc}^{\curlywedge}_{{\langle{r,f}\rangle}}(\mathcal{M}_{1})
$}\normalsize}\]
in the fiber mathematical context of structures 
$\mathrmbf{Struc}(\mathcal{S}_{2})$.
Hence,
a structure morphism satisfies the following conditions.
\begin{center}
{\footnotesize{\begin{tabular}{r@{\hspace{10pt}}c@{\hspace{10pt}}l}
\multicolumn{3}{l}{\textsf{list preservation}}
\\
\multicolumn{3}{c}{$r{\;\cdot\;}\sigma_{1}\;=\;\sigma_{2}{\;\cdot\;}{\scriptstyle\sum}_{f}$}
\\
\multicolumn{3}{c}{$k{\;\cdot\;}\tau_{2}\;=\;\tau_{1}{\;\cdot\;}{\scriptstyle\sum}_{g}$}
\\
\multicolumn{3}{l}{\textsf{infomorphisms}}
\\
$k_{1}{\;\models_{\mathcal{E}_{1}}\;}r(r_{2})$
&
\underline{iff}
&
$k(k_{1}){\;\models_{\mathcal{E}_{2}}\;}r_{2}$
\\
$y_{1}{\;\models_{\mathcal{A}_{1}}\;}f(x_{2})$
&
\underline{iff}
&
$g(y_{1}){\;\models_{\mathcal{A}_{2}}\;}x_{2}$
\end{tabular}}}
\end{center}
Thus,
a structure morphism
$\mathcal{M}_{2}\xrightleftharpoons{{\langle{r,k,f,g}\rangle}}\mathcal{M}_{1}$
(Fig.~\ref{fig:fole:struc:mor} in \S~\ref{sub:sub:sec:data:model:struc})
from source structure 
$\mathcal{M}_{2}={\langle{\mathcal{E}_{2},{\langle{\sigma_{2},\tau_{2}}\rangle},\mathcal{A}_{2}}\rangle}$
to target structure
$\mathcal{M}_{1}={\langle{\mathcal{E}_{1},{\langle{\sigma_{1},\tau_{1}}\rangle},\mathcal{A}_{1}}\rangle}$
is defined in terms of 
the hypergraph and classification morphisms between the source and target structure components (projections):
\begin{center}
{\footnotesize{$\begin{array}{
r@{\hspace{8pt}}
r@{\hspace{4pt}=\hspace{4pt}}
r@{\hspace{4pt}}
l@{\hspace{4pt}}
l@{\hspace{4pt}=\hspace{4pt}}
l}
\text{\emph{universe morphism}}
& \mathcal{U}_{2}
& \mathrmbfit{univ}(\mathcal{M}_{2})
& \xRightarrow{{\langle{k,g}\rangle}}
& \mathrmbfit{univ}(\mathcal{M}_{1})
& \mathcal{U}_{1}
\\
\text{\emph{schema morphism}}
& \mathcal{S}_{2}
& \mathrmbfit{sch}(\mathcal{M}_{2})
& \xRightarrow{{\langle{r,f}\rangle}}
& \mathrmbfit{sch}(\mathcal{M}_{1})
& \mathcal{S}_{1}
\\
\text{\emph{typed domain morphism}}
& \mathcal{A}_{2}
& \mathrmbfit{attr}(\mathcal{M}_{2})
& \xrightleftharpoons{{\langle{f,g}\rangle}}
& \mathrmbfit{attr}(\mathcal{M}_{1})
& \mathcal{A}_{1}
\\
\text{\emph{entity infomorphism}}
& \mathcal{E}_{2}
& \mathrmbfit{ent}(\mathcal{M}_{2})
& \xrightleftharpoons{{\langle{r,k}\rangle}}
& \mathrmbfit{ent}(\mathcal{M}_{1})
& \mathcal{E}_{1}
\end{array}$}}
\end{center}
Structure morphisms compose component-wise.
Let $\mathrmbf{Struc}$ denote the context of structures and structure morphisms.
(Fig.~\ref{fig:cxt:struc} in \S~\ref{sub:sub:sec:data:model:struc})
%

\newpage
\subsection{Universe Orientation.}\label{sub:sub:sec:univ:index}

As evident from the type-instance duality in 
Fig.~\ref{fig:fole:struc} and Fig.~\ref{fig:fole:struc:mor} of \S~\ref{sub:sec:fole:comps},
the fibered mathematical context $\mathrmbf{Struc}$ of structures and structure morphisms
can be developed from the dual standpoint
--- 
the Grothedieck construction of 
a universe indexed mathematical context.
This approach using universe indexing corresponds to the use of instance-set indexing in \S~\ref{sub:sub:sec:cls:fbr:cxt}.
\begin{enumerate}
\item 
For a fixed universe $\mathcal{U}$,
we define the structure fiber context $\mathrmbf{Struc}(\mathcal{U})$.
\item 
For a universe morphism
$\mathcal{U}_{2}\xLeftarrow{{\langle{k,g}\rangle}}\mathcal{U}_{1}$,
we define the structure fiber passage
$\mathrmbf{Struc}(\mathcal{U}_{2})\xrightarrow{\mathrmbfit{struc}^{\curlyvee}_{{\langle{k,g}\rangle}}}\mathrmbf{Struc}(\mathcal{U}_{1})$.
\rule[-8pt]{0pt}{10pt}
\item 
We define 
the fibered mathematical context $\mathrmbf{Struc}$ of structures and structure morphisms
to be the Grothedieck construction of 
the universe indexed mathematical context 
$\mathrmbf{Univ}^{\mathrm{op}}\xrightarrow{\mathrmbfit{struc}^{\curlyvee}}\mathrmbf{Cxt}$
\end{enumerate}
%

\paragraph{Fixed Universe.}

%
Let $\mathcal{U} = {\langle{K,\tau,Y}\rangle}\in\mathrmbf{Univ}$ be a fixed universe.
We define the notion of a morphism between two (fixed universe) $\mathcal{U}$-structures that respects the {\ttfamily ERA} data model.
In these morphisms,
the universe remains fixed,
but the attribute types (sorts) and entity types
are formally linked by maps that respect 
schema signature, typed domain extent and entity interpretation. 
A $\mathcal{U}$-structure morphism
$\mathcal{M}_{2}
\xrightleftharpoons{{\langle{r,f}\rangle}}
\mathcal{M}_{1}$
over a fixed universe $\mathcal{U} = {\langle{K,\tau,Y}\rangle}$
is a schema morphism  
$\mathrmbfit{sch}(\mathcal{M}_{2}) 
\xRightarrow{\;{\langle{r,f}\rangle}\;}
\mathrmbfit{sch}(\mathcal{M}_{1})$,
where
$\mathrmbfit{attr}(\mathcal{M}_{2})\xrightleftharpoons{{\langle{f,1_{Y}}\rangle}}\mathrmbfit{attr}(\mathcal{M}_{1})$
is an infomorphism in 
$\mathrmbf{Cls}(Y)$
over the value set $Y$, 
and
$\mathrmbfit{ent}(\mathcal{M}_{2})\xrightleftharpoons{{\langle{r,1_{K}}\rangle}}\mathrmbfit{ent}(\mathcal{M}_{1})$
is an infomorphism in 
$\mathrmbf{Cls}(K)$
over the key set $K$.
\footnote{Hence,
the source type domain is the inverse image of the target type domain
$\mathrmbfit{attr}(\mathcal{M}_{2}) 
= f^{-1}(\mathrmbfit{attr}(\mathcal{M}_{1}))
= \mathrmbfit{cls}^{\curlywedge}_{f}(\mathrmbfit{attr}(\mathcal{M}_{1}))$
and
the source entity classification is the inverse image of the target entity classification
$\mathrmbfit{ent}(\mathcal{M}_{2}) 
= r^{-1}(\mathrmbfit{ent}(\mathcal{M}_{1}))
= \mathrmbfit{cls}^{\curlywedge}_{r}(\mathrmbfit{ent}(\mathcal{M}_{1}))$.
By combining entity/attribute inverse image classifications,
the source structure $\mathcal{M}_{2}$ is the inverse image 
$\mathrmbfit{struc}^{\curlywedge}_{{\langle{r,f}\rangle}}(\mathcal{M}_{1}) 
={\langle{r^{-1}(\mathrmbfit{ent}(\mathcal{M}_{1})),{\langle{\sigma_{2},\tau}\rangle},f^{-1}(\mathrmbfit{attr}(\mathcal{M}_{1}))}\rangle}$
of the target structure $\mathcal{M}_{1}$
(\S~\ref{sub:sub:sec:sch:index}).}
Here,
$\mathcal{M}_{2}$ is called the reduct of $\mathcal{M}_{1}$ and
$\mathcal{M}_{1}$ is called the expansion of $\mathcal{M}_{2}$.
The two are bridged 
$\mathcal{M}_{2}=\mathrmbfit{struc}^{\curlywedge}_{{\langle{r,f}\rangle}}(\mathcal{M}_{1}) 
\xrightleftharpoons{{\langle{r,f}\rangle}}
\mathcal{M}_{1}$
by the structure morphism.
\footnote{When this definition is extended to formulas,
one gets the notion of an interpretation of first-order logic (extended to the many-sorted case)
given in (Barwise and Selman~\cite{barwise:seligman:97}).}
$\mathcal{U}$-structure morphisms compose component-wise.
Let $\mathrmbf{Struc}(\mathcal{U})$ 
denote the fiber context of structures
over the fixed universe $\mathcal{U}$.


\paragraph{Structure Fiber Passage.}

We define the indexed context 
$\mathrmbf{Univ}^{\mathrm{op}}\xrightarrow{\mathrmbfit{struc}^{\curlyvee}}\mathrmbf{Cxt}$.
Given a universe $\mathcal{U}$,
there is a fiber context of structures $\mathrmbf{Struc}(\mathcal{U})$ with that universe.
Given a universe morphism
$\mathcal{U}_{2}\xLeftarrow{{\langle{k,g}\rangle}}\mathcal{U}_{1}$,
there is a fiber passage of structures
$\mathrmbf{Struc}(\mathcal{U}_{2})\xrightarrow{\mathrmbfit{struc}^{\curlyvee}_{{\langle{k,g}\rangle}}}\mathrmbf{Struc}(\mathcal{U}_{1})$:
a structure 
$\mathcal{M}_{2}={\langle{\mathcal{E}_{2},{\langle{\sigma_{2},\tau_{2}}\rangle},\mathcal{A}_{2}}\rangle}\in\mathrmbf{Struc}(\mathcal{U}_{2})$
is mapped to a structure 
$\mathcal{M}_{1}=\mathrmbfit{struc}^{\curlyvee}_{{\langle{k,g}\rangle}}(\mathcal{M}_{2}) 
={\langle{k^{-1}(\mathcal{E}_{2}),{\langle{\sigma_{2},\tau_{1}}\rangle},g^{-1}(\mathcal{A}_{2})}\rangle}
\in\mathrmbf{Struc}(\mathcal{U}_{1})$.
The two are linked 
$\mathcal{M}_{2}
\xleftrightharpoons[{\langle{k,g}\rangle}]{\acute{\chi}_{\mathcal{M}_{2}}}
\mathrmbfit{struc}^{\curlyvee}_{{\langle{k,g}\rangle}}(\mathcal{M}_{2})$
by a structure morphism,
which is the $\mathcal{M}_{2}^{\text{th}}$ component of a bridge
$\mathrmbfit{inc}_{\mathcal{U}_{2}}
\xRightarrow{\;\,\acute{\chi}_{{\langle{k,g}\rangle}}\,}
\mathrmbfit{struc}^{\curlyvee}_{{\langle{k,g}\rangle}}{\;\circ\;}\mathrmbfit{inc}_{\mathcal{U}_{1}}$.
\begin{center}
\begin{tabular}{c@{\hspace{110pt}}c}
\begin{tabular}{c}
\setlength{\unitlength}{0.56pt}
\begin{picture}(120,120)(0,20)
\put(0,120){\makebox(0,0){\footnotesize{$\mathrmbf{Struc}(\mathcal{U}_{2})$}}}
\put(120,120){\makebox(0,0){\footnotesize{$\mathrmbf{Struc}(\mathcal{U}_{1})$}}}
\put(60,30){\makebox(0,0){\footnotesize{$\mathrmbf{Struc}$}}}
\put(60,138){\makebox(0,0){\scriptsize{$\mathrmbfit{struc}^{\curlyvee}_{{\langle{k,g}\rangle}}$}}}
\put(20,73){\makebox(0,0)[r]{\scriptsize{$\mathrmbfit{inc}_{\mathcal{U}_{2}}$}}}
\put(100,73){\makebox(0,0)[l]{\scriptsize{$\mathrmbfit{inc}_{\mathcal{U}_{1}}$}}}
\put(60,90){\makebox(0,0){\normalsize{$\xRightarrow{\;\acute{\chi}_{{\langle{k,g}\rangle}}}$}}}
\put(45,120){\vector(1,0){30}}
\put(10,105){\vector(2,-3){40}}
\put(110,105){\vector(-2,-3){40}}
\end{picture}
\end{tabular}
&
{{\begin{tabular}{c}
\setlength{\unitlength}{0.58pt}
\begin{picture}(180,180)(-40,-80)
\put(0,80){\makebox(0,0){\footnotesize{$k^{-1}(\mathcal{E}_{2})$}}}
\put(46,0){\makebox(0,0)[r]{\footnotesize{${\scriptstyle\sum}_{g}^{-1}(\mathrmbf{List}(\mathcal{A}_{2}))=
\mathrmbf{List}(g^{-1}(\mathcal{A}_{2}))$}}}
\put(-15,40){\makebox(0,0)[r]{\scriptsize{${\langle{\sigma_{2},\tau_{1}}\rangle}$}}}
\put(0,40){\makebox(0,0){\large{$\downdownarrows$}}}
\put(120,80){\makebox(0,0){\footnotesize{$\mathcal{E}_{2}$}}}
\put(120,0){\makebox(0,0){\footnotesize{$\mathrmbf{List}(\mathcal{A}_{2})$}}}
\put(105,40){\makebox(0,0)[r]{\scriptsize{${\langle{\sigma_{2},\tau_{2}}\rangle}$}}}
\put(120,40){\makebox(0,0){\large{$\downdownarrows$}}}
\put(65,80){\makebox(0,0){\footnotesize{$\xleftrightharpoons{{\langle{1_{R_{2}},k}\rangle}}$}}}
\put(65,-20){\makebox(0,0){\scriptsize{$\mathrmbf{List}(1_{X_{2}},g)$}}}
\put(65,0){\makebox(0,0){\normalsize{$\leftrightharpoons$}}}
\put(60,-60){\makebox(0,0){\footnotesize{$
\overset{\underbrace{\rule{70pt}{0pt}}}
{\mathrmbfit{struc}^{\curlyvee}_{{\langle{k,g}\rangle}}(\mathcal{M}_{2})
\xleftrightharpoons[{\langle{1_{R_{2}},k,1_{X_{2}},g}\rangle}]{{(\acute{\chi}_{{\langle{k,g}\rangle}})}_{\mathcal{M}_{2}}} 
\mathcal{M}_{2}\;\;\;\;\;\;}$}}}
\end{picture}
\end{tabular}}}
\end{tabular}
\end{center}

\paragraph{Multiple Universes.}

The Grothedieck construction of 
the universe indexed mathematical context
$\mathrmbf{Univ}^{\mathrm{op}}\xrightarrow{\mathrmbfit{struc}^{\curlyvee}}\mathrmbf{Cxt}$
is the fibered mathematical context $\mathrmbf{Struc}$ of structures and structure morphisms.
A structure $\mathcal{M}$ 
is as described in \S~\ref{sub:sub:sec:data:model:struc}
(Fig.~\ref{fig:fole:struc}).
A structure morphism
$\mathcal{M}_{2}\xrightleftharpoons{{\langle{r,k,f,g}\rangle}}\mathcal{M}_{1}$
from source structure $\mathcal{M}_{2}$
to target structure $\mathcal{M}_{1}$
consists of a universe morphism 
$\mathcal{U}_{2}\xLeftarrow{{\langle{k,g}\rangle}}\mathcal{U}_{1}$
to target structure $\mathcal{U}_{2}$
from source structure $\mathcal{U}_{1}$
and a morphism 
\[\mbox{\footnotesize{$
\mathrmbfit{struc}^{\curlyvee}_{{\langle{k,g}\rangle}}(\mathcal{M}_{2})
={\langle{k^{-1}(\mathcal{E}_{2}),{\langle{\sigma_{2},\tau_{1}}\rangle},g^{-1}(\mathcal{A}_{2})}\rangle}
\xrightleftharpoons{{\langle{r,f}\rangle}}
{\langle{\mathcal{E}_{1},{\langle{\sigma_{1},\tau_{1}}\rangle},\mathcal{A}_{1}}\rangle}=\mathcal{M}_{1}
$}\normalsize}\]
in the fiber mathematical context of structures 
$\mathrmbf{Struc}(\mathcal{U}_{1})$.
Hence,
a structure morphism satisfies the following conditions.
\begin{center}
{\footnotesize{\begin{tabular}{r@{\hspace{10pt}}c@{\hspace{10pt}}l}
\multicolumn{3}{l}{\textsf{list preservation}}
\\
\multicolumn{3}{c}{$r{\;\cdot\;}\sigma_{1}\;=\;\sigma_{2}{\;\cdot\;}{\scriptstyle\sum}_{f}$}
\\
\multicolumn{3}{c}{$k{\;\cdot\;}\tau_{2}\;=\;\tau_{1}{\;\cdot\;}{\scriptstyle\sum}_{g}$}
\\
\multicolumn{3}{l}{\textsf{infomorphisms}}
\\
$k_{1}{\;\models_{\mathcal{E}_{1}}\;}r(r_{2})$
&
\underline{iff}
&
$k(k_{1}){\;\models_{\mathcal{E}_{2}}\;}r_{2}$
\\
$y_{1}{\;\models_{\mathcal{A}_{1}}\;}f(x_{2})$
&
\underline{iff}
&
$g(y_{1}){\;\models_{\mathcal{A}_{2}}\;}x_{2}$
\end{tabular}}}
\end{center}
Thus,
a structure morphism
$\mathcal{M}_{2}\xrightleftharpoons{{\langle{r,k,f,g}\rangle}}\mathcal{M}_{1}$
(Fig.~\ref{fig:fole:struc:mor} in \S~\ref{sub:sub:sec:data:model:struc})
from source structure 
$\mathcal{M}_{2}={\langle{\mathcal{E}_{2},{\langle{\sigma_{2},\tau_{2}}\rangle},\mathcal{A}_{2}}\rangle}$
to target structure
$\mathcal{M}_{1}={\langle{\mathcal{E}_{1},{\langle{\sigma_{1},\tau_{1}}\rangle},\mathcal{A}_{1}}\rangle}$
is defined in terms of 
the hypergraph and classification morphisms between the source and target structure components (projections):
\begin{center}
{\footnotesize{$\begin{array}{
r@{\hspace{8pt}}
r@{\hspace{4pt}=\hspace{4pt}}
r@{\hspace{4pt}}
l@{\hspace{4pt}}
l@{\hspace{4pt}=\hspace{4pt}}
l}
\text{\emph{universe morphism}}
& \mathcal{U}_{2}
& \mathrmbfit{univ}(\mathcal{M}_{2})
& \xRightarrow{{\langle{k,g}\rangle}}
& \mathrmbfit{univ}(\mathcal{M}_{1})
& \mathcal{U}_{1}
\\
\text{\emph{schema morphism}}
& \mathcal{S}_{2}
& \mathrmbfit{sch}(\mathcal{M}_{2})
& \xRightarrow{{\langle{r,f}\rangle}}
& \mathrmbfit{sch}(\mathcal{M}_{1})
& \mathcal{S}_{1}
\\
\text{\emph{typed domain morphism}}
& \mathcal{A}_{2}
& \mathrmbfit{attr}(\mathcal{M}_{2})
& \xrightleftharpoons{{\langle{f,g}\rangle}}
& \mathrmbfit{attr}(\mathcal{M}_{1})
& \mathcal{A}_{1}
\\
\text{\emph{entity infomorphism}}
& \mathcal{E}_{2}
& \mathrmbfit{ent}(\mathcal{M}_{2})
& \xrightleftharpoons{{\langle{r,k}\rangle}}
& \mathrmbfit{ent}(\mathcal{M}_{1})
& \mathcal{E}_{1}
\end{array}$}}
\end{center}
Structure morphisms compose component-wise.
Let $\mathrmbf{Struc}$ denote the context of structures and structure morphisms.
(Fig.~\ref{fig:cxt:struc} in \S~\ref{sub:sub:sec:data:model:struc})
%


\begin{proposition}\label{prop:struc:mor:dual:facts}
(compare Prop.~\ref{prop:cls:mor:dual:facts})
Any structure morphism 
$\mathcal{M}_{2}\xrightleftharpoons{{\langle{r,k,f,g}\rangle}}\mathcal{M}_{1}$,
with schema morphism projection 
$\mathcal{S}_{2}\xRightarrow{{\langle{r,f}\rangle}}\mathcal{S}_{1}$
and universe morphism projection 
$\mathcal{U}_{2}\xLeftarrow{{\langle{k,g}\rangle}}\mathcal{U}_{1}$,
has dual factorizations
(see the diagram below).
\begin{center}
{{\begin{tabular}{c}
\setlength{\unitlength}{0.64pt}
\begin{picture}(120,140)(0,-35)
\put(60,90){\makebox(0,0){\footnotesize{$\underset{\textstyle{\overbrace{\rule{80pt}{0pt}}}}{\mathrmbf{Struc}(\mathcal{S}_{2})}$}}}
\put(60,-30){\makebox(0,0){\footnotesize{$\overset{\textstyle{\underbrace{\rule{80pt}{0pt}}}}{\mathrmbf{Struc}(\mathcal{U}_{1})}$}}}
\put(0,60){\makebox(0,0){\footnotesize{$\mathcal{M}_{2}$}}}
\put(120,60){\makebox(0,0){\footnotesize{$\mathrmbfit{struc}^{\curlywedge}_{{\langle{r,f}\rangle}}(\mathcal{M}_{1})$}}}
\put(120,0){\makebox(0,0){\footnotesize{$\mathcal{M}_{1}$}}}
\put(0,0){\makebox(0,0){\footnotesize{$\mathrmbfit{struc}^{\curlyvee}_{{\langle{k,g}\rangle}}(\mathcal{M}_{2})$}}}
\put(-7,29){\makebox(0,0)[r]{\scriptsize{${\langle{k,g}\rangle}$}}}
\put(80,-12){\makebox(0,0){\scriptsize{${\langle{r,f}\rangle}$}}}
\put(40,72){\makebox(0,0){\scriptsize{${\langle{k,g}\rangle}$}}}
\put(127,30){\makebox(0,0)[l]{\scriptsize{${\langle{r,f}\rangle}$}}}
%
%
\put(40,57){\makebox(0,0){\large{$\xrightleftharpoons{\;\;\;\;\;\;\;\;}$}}}
\put(80,-3){\makebox(0,0){\large{$\xrightleftharpoons{\;\;\;\;\;\;\;\;}$}}}
\put(118,30){\makebox(0,0){\large{$\upharpoonleft$}}}
\put(122,30){\makebox(0,0){\large{$\downharpoonright$}}}
\put(-2,30){\makebox(0,0){\large{$\upharpoonleft$}}}
\put(2,30){\makebox(0,0){\large{$\downharpoonright$}}}
\end{picture}
\end{tabular}}}
\end{center}
The top-right factorization
(corresponding to the schema orientation of \S~\ref{sub:sub:sec:sch:index})
consists of 
a $\mathrmbf{Struc}(\mathcal{S}_{2})$-morphism 
$\mathcal{M}_{2}\xrightleftharpoons{{\langle{k,g}\rangle}}\mathrmbfit{struc}^{\curlywedge}_{{\langle{r,f}\rangle}}(\mathcal{M}_{1})$
and 
the $\mathcal{M}_{1}^{\text{th}}$ component 
$\mathrmbfit{struc}^{\curlywedge}_{{\langle{r,f}\rangle}}(\mathcal{M}_{1})\xrightleftharpoons{{\langle{r,f}\rangle}}\mathcal{M}_{1}$
of a bridge
$\mathrmbfit{struc}^{\curlywedge}_{{\langle{r,f}\rangle}}{\;\circ\;}\mathrmbfit{inc}_{\mathcal{S}_{2}}
\xRightarrow{\,\grave{\chi}_{{\langle{r,f}\rangle}}\;\,}\mathrmbfit{inc}_{\mathcal{S}_{1}}$.
The left-bottom factorization 
(corresponding to the universe orientation of \S~\ref{sub:sub:sec:univ:index})
consists of 
the $\mathcal{M}_{2}^{\text{th}}$ component 
$\mathcal{M}_{2}\xrightleftharpoons{{\langle{k,g}\rangle}}\mathrmbfit{struc}^{\curlyvee}_{{\langle{k,g}\rangle}}(\mathcal{M}_{2})$
of a bridge
$\mathrmbfit{inc}_{\mathcal{U}_{2}}
\xRightarrow{\,\acute{\chi}_{{\langle{k,g}\rangle}}\;}
\mathrmbfit{struc}^{\curlyvee}_{{\langle{k,g}\rangle}}{\;\circ\;}\mathrmbfit{inc}_{\mathcal{U}_{1}}$
and
a $\mathrmbf{Struc}(\mathcal{U}_{1})$-morphism
$\mathrmbfit{struc}^{\curlyvee}_{{\langle{k,g}\rangle}}(\mathcal{M}_{2})\xrightleftharpoons{{\langle{r,f}\rangle}}\mathcal{M}_{1}$.
%
\end{proposition}





\end{document}